\documentclass[submission,copyright,creativecommons]{eptcs}
\providecommand{\event}{GandALF 2016} 
\usepackage{breakurl}             

\usepackage{graphicx}
\usepackage{amsmath}
\usepackage{amssymb}
\usepackage{underscore}

\usepackage{url}
\usepackage{caption}

\newcommand{\truth}{\top}
\newcommand{\falsity}{\bot}
\newcommand{\atoms}{AP}
\newcommand{\tick}{\surd}
\newcommand{\cross}{\times}
\newcommand{\prune}{PRUNE}
\newcommand{\prunez}{PRUNE$_0$}

\newcommand{\ignore}[1]{}
\newcommand{\longversion}[1]{}
\newcommand{\shortversion}[1]{#1}

 \newcommand{\webpage}{\url{http://staffhome.ecm.uwa.edu.au/~00054620/research/Online/ltlsattab.html}}
 \newcommand{\techreport}{DBLP:journals/corr/Reynolds16}
 \newcommand{\ijcai}{DBLP:conf/ijcai/BertelloGMR16}
 
 \newenvironment{lemma}{LEMMA:}{}
 \newenvironment{proof}{PROOF:}{} 

\title{A New Rule for LTL Tableaux}
\author{Mark Reynolds
\institute{The University of Western Australia\thanks{The work was partially supported
by the Australian Research Council.
The author would also like to 
thank staff, students
and Prof. Angelo Montanari 
of Udine,
for interesting and useful discussions
on the development of this approach:
and a chance to check
how intuitive the new approach is.
}\\
Perth, Australia}
\email{mark.reynolds@uwa.edu.au}
}
\def\titlerunning{New Rule for LTL}
\def\authorrunning{M. Reynolds}
\begin{document}
\maketitle

\begin{abstract}
Propositional linear time temporal logic (LTL) is the standard temporal logic for computing applications and many reasoning techniques and tools have been developed for it. Tableaux for deciding satisfiability have existed since the 1980s. However, the tableaux for this logic do not look like traditional tree-shaped tableau systems and their processing is often quite complicated. 

In this paper, 
we introduce a novel style of tableau rule
which supports a new simple traditional-style tree-shaped tableau for LTL.
We prove that it is sound and complete. As well as being simple to understand, to introduce to students and to use, it is also simple to implement and is competitive against state of the art systems. 
It is particularly suitable for parallel implementations.
\end{abstract}


\section{Introduction}
\label{sec:intro}

Propositional linear time temporal logic,
LTL, is important for hardware and software verification \cite{DBLP:conf/spin/RozierV07}.
LTL satisfiability checking (LTLSAT) is receiving renewed interest
with advances in computing power, several industry ready tools,
some new theoretical techniques,
studies of the relative merits of different approaches,
implementation competitions, and benchmarking:
\cite{Goranko2010113,VSchuppanLDarmawan-ATVA-2011,DBLP:conf/spin/RozierV07}.
Common techniques
include
automata-based approaches 
\cite{VaW94}
and
resolution
\cite{DBLP:journals/aicom/LudwigH10}
as well as tableaux \cite{Gou89,Wol85,DBLP:conf/cav/KestenMMP93}.
Each type of approach
has its own advantages and disadvantages
and each can be competitive
at the industrial scale
(albeit within the limits of what may
be achieved with PSPACE complexity).
State of the art systems
such as \cite{RV11}
often incorporate 
a variety of previous approaches working in parallel
and thus
piggy-back on the
fastest existing tableaux and fastest
versions of other approaches.

Traditionally tableaux 
present as tree-shaped
\cite{Beth55,Smu68,Girle00}.
For temporal logic, tableaux tend to be
{\em declarative}, which 
means that the
definition of a node
is as a set of formulas
(so a particular set appears 
at most once), and the successor relation
is determined by those formulas;
and they tend to be graph-shaped.
Figure~\ref{fig:cgh97tab}
shows a typical graph-style tableau,
a more general {\em graph} rather than tree-shaped
(this from \cite{CGH97}).
In general, 
these tableau constructions
need the whole graph to be present
 before
 a second phase of discarding nodes takes place.


\begin{figure}
\centering
\begin{minipage}{.5\textwidth}
  \centering
\includegraphics[width=6cm,trim= 4cm 8.5cm 4cm 9cm,clip=true]{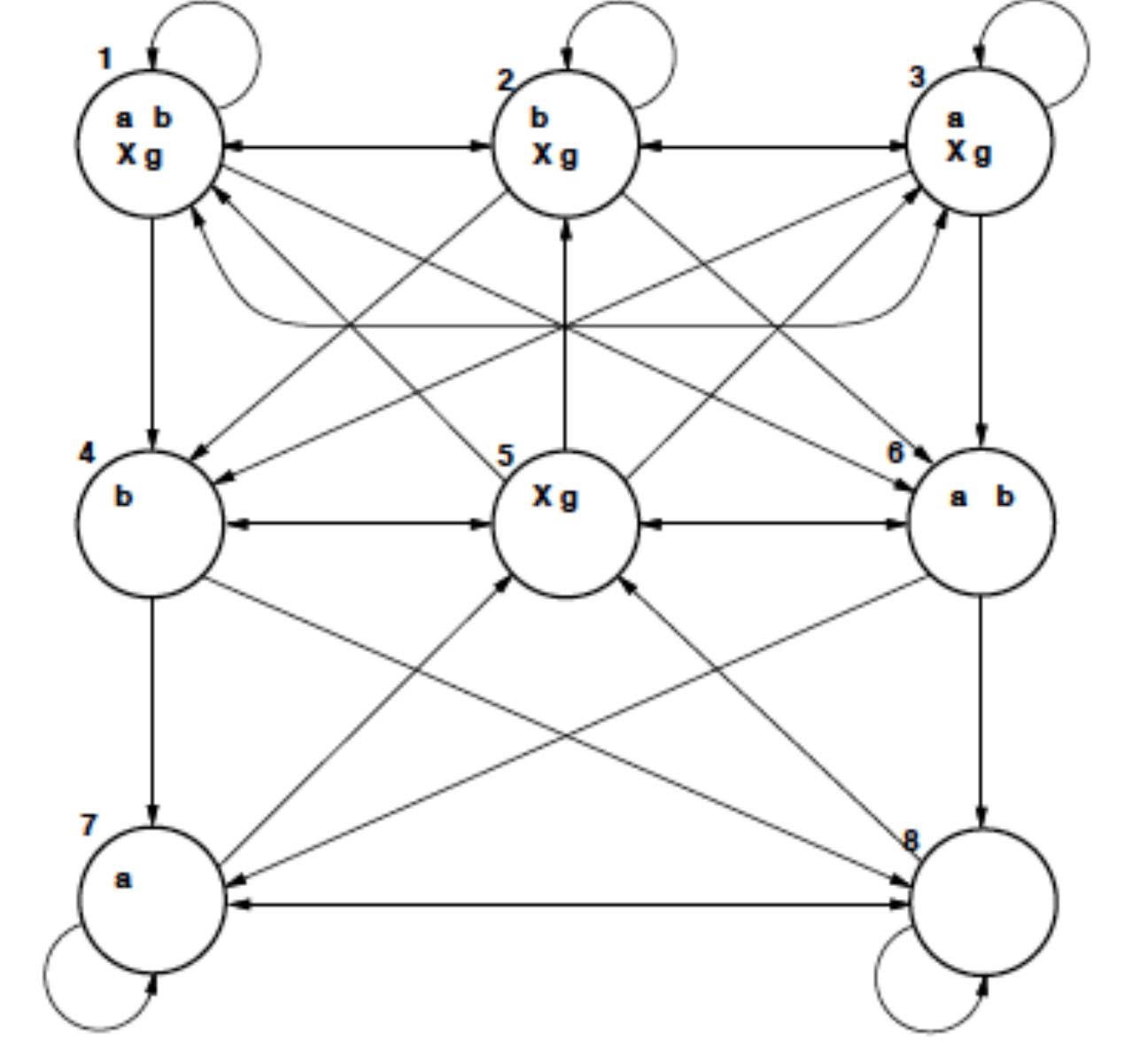}
  \captionof{figure}{A graph-shaped tableau for $g= a U b$ from \cite{CGH97}}
  \label{fig:cgh97tab}
\end{minipage}%
\begin{minipage}{.5\textwidth}
  \centering
\includegraphics[width=6cm,trim= 4cm 11cm 4cm 11cm,clip=true]{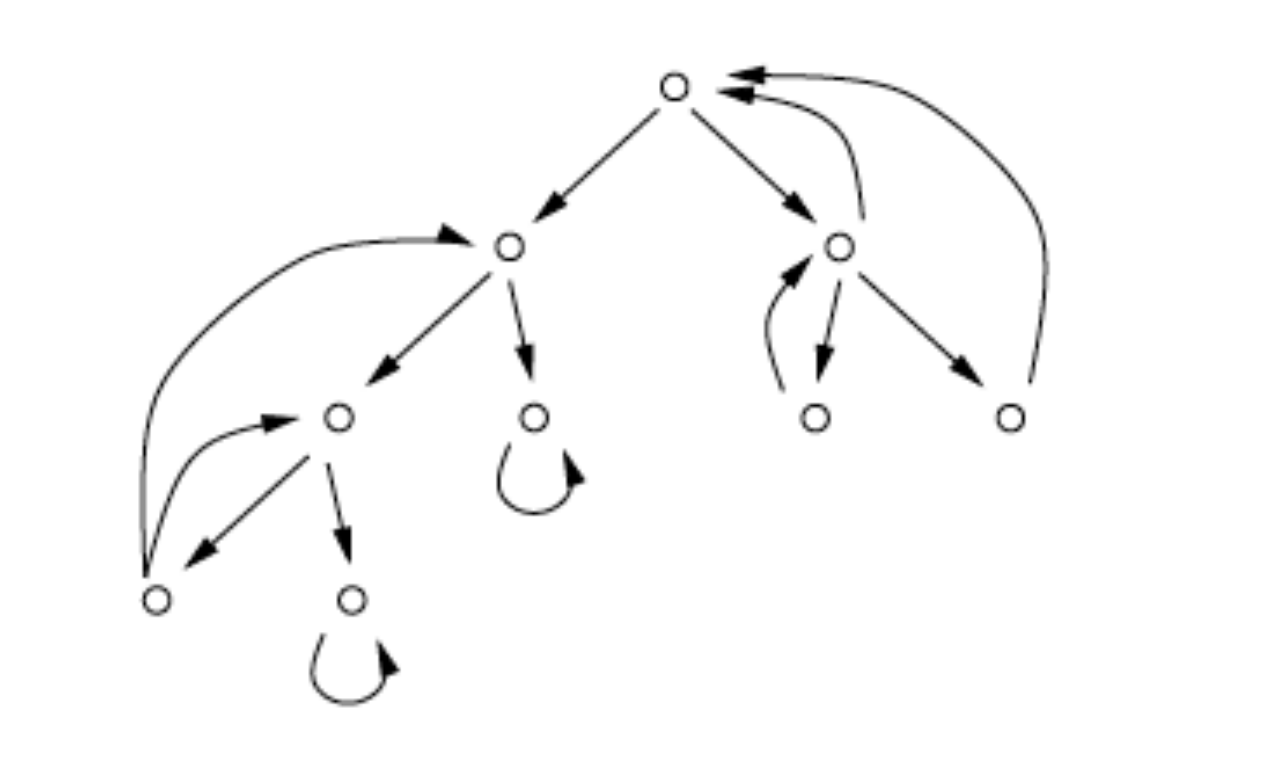}
  \captionof{figure}{A tree-shaped tableau from \cite{Sch98}}
  \label{fig:schwetab}
\end{minipage}
\end{figure}

Figure~\ref{fig:schwetab}
shows the tree-style tableau
from \cite{Schwe98}
which is unusual:
now nodes are independent objects arranged in a tree
and although they are labelled by sets of formulas,
the same label may appear on different nodes.
This 
 stands out
 in being tree-shaped (not a more general graph),
 and in being one-pass, i.e.,
 not relying on a two-phase building and pruning process.
 It also stands out in speed (i.e., faster than other tableau approaches)
 \cite{Goranko2010113} 
 and is thus the state of the art in tableau
 approaches to LTL.
 Construction of tree-shaped tableaux
 can be local as we do not need to check if the
 same label has already been constructed 
 somewhere else.
 However,
 there are still elements
 of communication between separate branches
 and a slightly complicated
 annotation of nodes with depth measures that
 needs to be managed
 as it feeds in to the tableau rules.
So it is not in the traditional 
style of classical tableaux \cite{Mor12}.

The following is a new simpler tableau for LTL.
While
many of the rules are unsurprising 
given earlier LTL tableaux
\cite{Wol85,Schwe98},
the new rule to curtail the construction
of repetitive loops is interesting.
Loop checking is a common concern in tableau 
approaches modal logics and there have been 
many strange rules proposed.
There is a loop checker (for a branching logic) in \cite{BMP83}
which cuts construction after a label appears three times:
but this would cause incompleteness in LTL.
For logic S4, \cite{HSZ96}
require that boxes must accumulate.
For LTL we have loop checkers
in \cite{DBLP:journals/entcs/GaintzarainHLN08}
and \cite{BrL08}
but these have not proven practical 
because they involve collecting unwieldy annotations
(of conjuncts from a doubly exponentially-sized
closure set)
or sets of sets of contextual formulas
(see \cite{Wid10} for a
more detailed account).
The new PRUNE rule presented here
is instead simple and very intuitive:
don't return to a label
for a third time unless
there is progress on fulfilling eventualities
after the second time.

Perhaps the most similar previous work 
is that in 
\cite{BEM96}
which presents a tableau for
linear time $\mu$-calculus
in which there are several rules
operating on the sequence of labels in a branch.
One rule called (v)
classifies a branch as successful if 
a label is repeated three times with the same
fixed point being unwound first in each
intervening interval.
This corresponds roughly to
having the same eventuality postponed first
in each interval: so not quite as
easy to define or check in temporal logic.
As far as the author knows this tableau
has neither been implemented, nor translated to work with LTL.

Thus
the PRUNE rule is completely novel and a surprisingly simple way to curtail
repetitive branch extension. It may be applicable
in other contexts.
The overall tableau is thus novel
and it is
unique in that it is wholly traditional in style (labels are sets
of formulas), tree shaped tableau construction,
no extraneous recording of calculated values, contexts
or annotations, just looking at the labels.
It is also unique in allowing
separate parallel development of branches.

Because of the tree shape, the tableau search
allows completely independent
searching down separate branches
and so lends itself to parallel
computing.
In fact, this approach is
``embarrassingly parallel''
\cite{Fos95}.
Thus there is also potential for quantum
implementations.
Developing a parallel implementation is ongoing work,
however:
even though only formula set labels need to be recorded
down a branch,
there is a need to 
work out an efficient way for 
memory of the sequence of labels to be managed
below nodes where two branches separate.


Solid experimental work in 
\cite{\ijcai}
(see section~\ref{sec:complex})
shows that the new tableau is competitive
but that is not the task of this paper.

The interesting completeness and termination of the tableau search is
proved here. The proofs are mostly straightforward.
However, the completeness proof with the PRUNE rule
is interesting, new and quite complicated.
The main task of this paper is to present that
proof.

In this paper
we briefly remind the reader of the well-known syntax and semantics for LTL
in Section~\ref{sec:synsem},
describe our tableau approach 
in 
general terms
in Section~\ref{sec:tab},
present the rules in Section~\ref{sec:rules},
comment on the use of the rules
and provide some motivation for our approach
in Section~\ref{sec:motiv},
\longversion{
prove soundness section~\ref{sec:sound},
prove completeness section~\ref{sec:complete},
}
\shortversion{
discuss soundness and
prove completeness in Section~\ref{sec:correctness},
}
and briefly discuss
 implementation issues in Section~\ref{sec:complex},
before a conclusion
in Section~\ref{sec:concl}.
Full versions of the (short) proofs can be found
in an online technical report
\cite{\techreport}.

\section{Syntax and Semantics}
\label{sec:synsem}

We assume a countable set $\atoms$
of propositional atoms,
or atomic propositions.

A (transition) structure
is a triple $(S,R,g)$
with
$S$ a finite set of states,
$R \subseteq S \times S$
a binary relation
called the transition relation
and labelling $g$ tells us which atoms are true
at each state: for each $s \in S$,
$g(s) \subseteq \atoms$.
Furthermore, $R$ is assumed to be serial: every state has at least one successor
$\forall x \in S.  \exists y \in S \mbox{ s.t.} (x,y) \in R$.

Given a structure $(S,R,g)$,
an $\omega$-long sequence of states
$\langle s_0, s_1, s_2, ... \rangle$
from $S$ is a {\em fullpath}
(through $(S,R,g)$)
iff
for each $i$,
$(s_i,s_{i+1}) \in R$.
If $\sigma= \langle s_0, s_1, s_2, ... \rangle$
is a fullpath
then
we write
$\sigma_i=s_i$,
$\sigma_{\geq j}= \langle s_j, s_{j+1}, s_{j+2}, ... \rangle$
(also a fullpath).

The (well formed) formulas of LTL
include the atoms and
if $\alpha$ and $\beta$ are formulas then so are
$\neg \alpha$, $\alpha \wedge \beta$,
$X \alpha$, and $\alpha U \beta$.
We will also include some formulas
built using other connectives
that are often presented as abbreviations instead.
However, before detailing them
we present the semantic clauses.

Semantics defines
truth of formulas
on a fullpath through a structure.
Write $M, \sigma \models \alpha$ iff the
formula $\alpha$ is true of the fullpath $\sigma$
in the structure $M=(S,R,g)$ defined recursively by:\\
\begin{tabular}{lll}
$M, \sigma \models p$ & iff & $p \in g(\sigma_0)$, for $p \in \atoms$;\\
$M, \sigma \models \neg \alpha$ & iff &
$M, \sigma \not \models \alpha$;\\
$M, \sigma \models \alpha \wedge \beta$ & iff &
$M, \sigma \models \alpha$ and 
$M, \sigma \models \beta$;\\
$M, \sigma \models X \alpha$ & iff &
$M, \sigma_{\geq 1} \models \alpha$; and\\
$M, \sigma \models \alpha U \beta$ & iff &   
there is some $i \geq 0$ s.t.
$M, \sigma_{\geq i} \models \beta$ and
for each $j$,\\&&
if $0 \leq j < i$ then
$M, \sigma_{\geq j}  \models \alpha$.\\ 
\end{tabular}

Standard abbreviations in LTL
include the
classical
$\truth \equiv p \vee \neg p$,
$\falsity \equiv \neg \truth$,
$\alpha \vee \beta \equiv \neg ( \neg \alpha \wedge \neg \beta)$,
$\alpha \rightarrow \beta \equiv \neg \alpha \vee \beta$,
$\alpha \leftrightarrow \beta \equiv
( \alpha \rightarrow \beta ) \wedge ( \beta \rightarrow \alpha)$.
We also have the temporal:
$F \alpha \equiv (\truth U \alpha)$,
$G \alpha \equiv \neg F ( \neg \alpha)$
read as eventually and  always respectively.

A formula $\alpha$ is {\em satisfiable} iff there
is some structure $(S,R,g)$
with some fullpath
$\sigma$ through it such that
$(S,R,g), \sigma \models \alpha$.
A formula is {\em valid} iff 
for all structures $(S,R,g)$
for all fullpaths
$\sigma$ through $(S,R,g)$ we have
$(S,R,g), \sigma \models \alpha$.
A formula is valid iff its negation 
is not satisfiable.

For example,
$\truth$, $p$, $Fp$, $p \wedge Xp \wedge F \neg p$,
$Gp$
are each satisfiable.
However,
$\falsity$, $p \wedge \neg p$,
$Fp \wedge G \neg p$,
$p \wedge G(p \rightarrow Xp) \wedge F \neg p$
are each not satisfiable.

We will fix a particular formula,
$\phi$ say,
and describe
how a tableau for $\phi$ is built
and
how that decides the
satisfiability
or otherwise, of $\phi$.
We will use other formula names
such as $\alpha$, $\beta$, e.t.c.,
to indicate
arbitrary formulas
which are used in labels
in the tableau for $\phi$.


\section{General Idea of the Tableau}
\label{sec:tab}

The tableau for $\phi$ is a tree of nodes
(going down the page from a root)
each labelled by 
a set of formulas.
To lighten the notation,
when we present a tableau in a diagram
we will omit the braces $\{ \}$
around the sets which form labels.
The root is labelled $\{ \phi\}$.


Each node has 0, 1 or 2 children:
it is the parent of its children.
A node is called a leaf if it has 0 children.
A leaf determines a branch,
being itself, its parent, its parent's parent e.t.c..
A leaf may be
crossed ($\times$),
indicating its branch has failed,
or ticked ($\surd$),
indicating the branch is successful.
Otherwise, a leaf indicates an unfinished branch
and having an unfinished branch 
means that the tableau is unfinished.
In that case there will be a way to
use one if the rules below
to extend the branch and the tableau.

The whole tableau is successful if there is a successful branch.
This indicates a ``yes'' answer to the satisfiability of $\phi$.
It is failed if all branches are failed: indicating ``no''.
Otherwise it is unfinished.
Note that you can stop the algorithm,
and report success if you tick a 
leaf even if other branches have not reached
a tick or cross yet.

A small set of tableau rules (see below) determine
whether a node has one or two children
or whether to cross or tick it.
This depends on the label of the parent,
and also, for some rules, on labels
on ancestor nodes,
higher up the branch towards the root.
The rule also determines
the labels on the children.

\newcommand{\horizontalTransition}{$\rightarrow \hspace{-0.45cm} / \hspace{-0.1cm} / \hspace{0.2cm}$}
\newcommand{\verticalTransition}{$\downarrow \hspace{-0.225cm} =$}

The parent-child relation
is indicated by a vertical arrow in some diagrams
but otherwise just by vertical alignment.
However, to indicate use of one particular rule
(coming up below)
called the TRANSITION rule
we will 
use a vertical arrow (\verticalTransition)
with two strikes across it,
or just an equals sign.

A node label may be the empty set, although
it then can be immediately ticked by rule EMPTY below.

A formula which is
an atomic proposition,
a negated atomic proposition
or of the form
$X \alpha$ or $\neg X \alpha$
is called {\em elementary}.
If a node label is non-empty
and there are no direct contradictions,
that is no $\alpha$ and $\neg \alpha$
amongst the formulas in the label,
and every formula it contains is elementary then
we call the label (or the node) {\em poised}.

\newcommand{\cupdot}{\mathbin{\mathaccent\cdot\cup}}

Most of the rules {\em consume} formulas.
That is,
the parent may have a label
$\Gamma = \Delta \cupdot \{ \alpha \}$,
where
$\cupdot$ is disjoint union,
and a child may have a label
$\Delta \cup \{ \gamma \}$ so that
$\alpha$ has been removed,
or consumed.
In diagrams, if it is useful we sometimes indicate such a
formula, known as a {\em pivot} formula, by underlining it.

See the
$\neg p \wedge X \neg p \wedge (q Up)$
example given in
Figure~\ref{fig:eg9}
of a simple successful tableau.
Note that if one was building, or searching,
the tableau in a depth-first left-to-right
manner then the tableau 
could be judged as successful
after the left-most ticked branch was found
and we could terminate the search 
without constructing 
quite as much as shown.
However, we are not assuming that tableaux
need to be constructed in that order.

\begin{figure}
\centering
\includegraphics[width=8cm,trim= 5cm 15.5cm 5cm 4.55cm,clip=true]{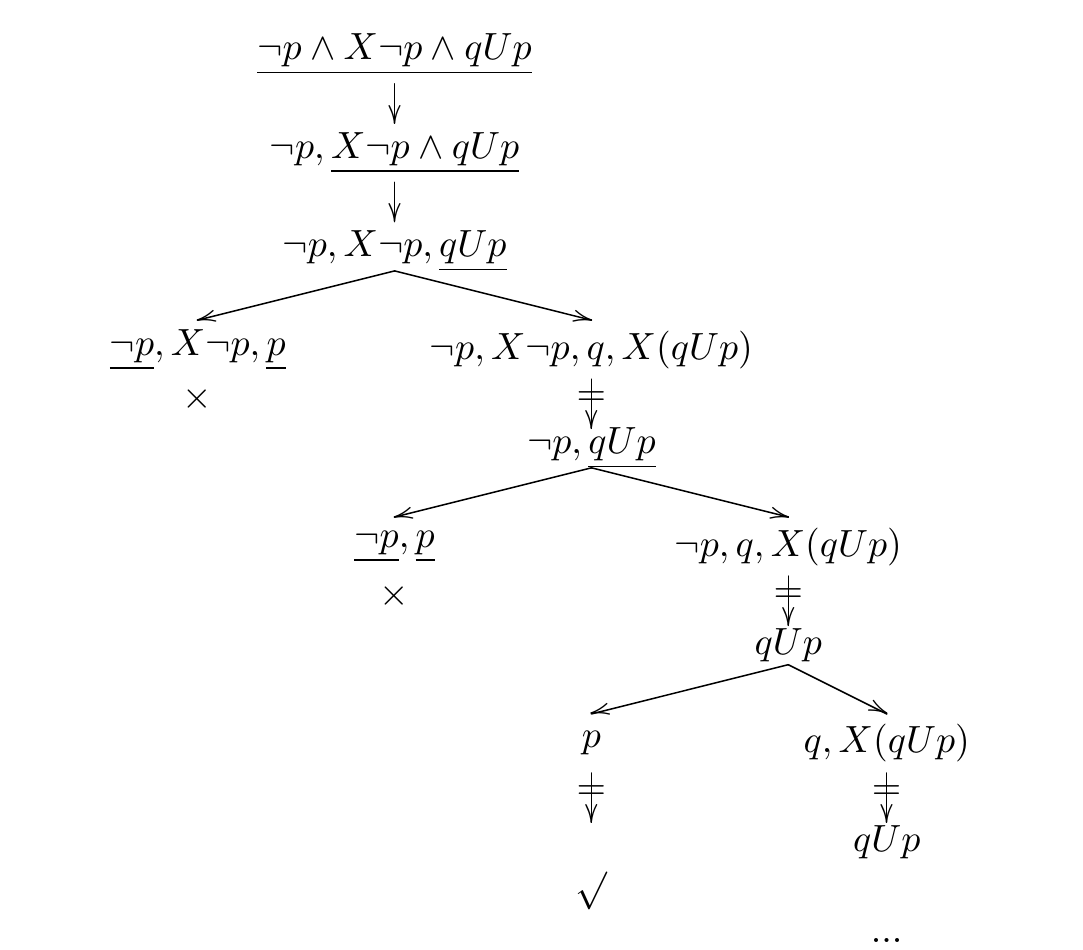}
\caption{$\neg p \wedge X \neg p \wedge (q Up)$}
\label{fig:eg9}
\end{figure}

As usual, a tableau node $x$ is an ancestor of a node $y$
precisely when $x=y$ or $x$ is a parent of $y$
or a parent of a parent of $y$, e.t.c.
Then $y$ is a descendent of $x$
and we write $x \leq y$.
Node $x$ is a proper ancestor of $y$,
written $x < y$, iff
it is an ancestor and $x \neq y$.
Similarly proper descendent.
When we say that a node $y$ is between
node $x$ and its descendent $z$,
$x \leq y \leq z$, then we mean that
$x$ is an ancestor of $y$ and
$y$ is an ancestor of $z$.

A formula of the form $X(\alpha U \beta)$ 
(or $X F \beta$ or $X \neg G  \beta'$)
appearing in a poised
label of a node $m$, say, also plays an important role.
We will call such a formula
an {\em $X$-eventuality} because
$\alpha U \beta$ (or $F \beta$ or $\neg G \beta'$) is
often called an {\em eventuality},
and its truth depends on $\beta$  (or $\beta=\neg \beta'$ in the $G$ case) being eventually
true in the future (if not present).
If the formula $\beta$ appears 
in the label of a descendent node $n$ of $m$
then we say that the $X$-eventuality
at $m$
has been
{\em fulfilled} by $n$ by $\beta$ being satisfied there.
 
\section{Rules}
\label{sec:rules}

There are twenty-five rules altogether. We would only need ten for the 
minimal LTL-language, but recall that
we are treating the usual abbreviations as
first-class symbols, so they each need
a pair of rules.
In this conference paper we skip
the rules for abbreviations.

Most of the rules are what we call {\em static} rules.
They tell us about 
formulas that may be true at a single state in a model.
They 
determine the number,
$0$, $1$ or $2$, of child nodes
and the labels on those nodes
from the label on the current parent node
without reference to any other labels.
These rules are unsurprising
to anyone familiar with any of the 
previous LTL tableau approaches.

To save repetition of wording we use
an abbreviated notation for presenting
each rule: the rule $A / B$
relates the parent label $A$
to the child labels $B$.
The parent label
is a set of formulas.
The child labels
are given as either
a $\tick$ representing the leaf of a successful branch, 
a $\cross$ representing the leaf of a failed branch, a 
single set being the label on the single child
or a pair of sets
separated by a vertical bar $|$ being the respective
labels on a pair of child nodes.

Thus, for example, the $U$-rule, means
that if a node is labelled 
$\{ \alpha U \beta \} \cupdot \Delta$,
if we choose to use the
$U$-rule and
if we choose to decompose $\alpha U \beta$ using the rule
then
the node will have two children labelled
$\Delta  \cup \{  \beta\}$
and 
$\Delta  \cup \{ \alpha,  X(\alpha U \beta)\}$
respectively.

Often, several different rules may be applicable to a
node with a certain label.
If another applicable rule is chosen,
or another formula is chosen to be decomposed
by the same rule,
then the child labels may be different.
We discuss this non-determinism later.

The four (static)  termination rules:\\
\begin{tabular}{lllll}
{EMPTY-rule:} &
$\{  \}  \;/\; \tick$; & $\;\;$ &
{$\falsity$-rule:} &
$\{ \falsity \} \cupdot \Delta \;/\; \times$;\\
{CONTRADICTION-rule:} &
$\{ \alpha, \neg \alpha \} \cupdot \Delta \;/\; \times$;& $\;\;$ &
{$\neg \truth$-rule:} &
$\{ \neg \truth \} \cupdot \Delta \;/\; \times$.\\
\end{tabular}

These are the positive static rules:\\
\begin{tabular}{ll}
{$\truth$-rule:} &
$\{ \truth \} \cupdot \Delta \;/\; \Delta$.\\
{$\wedge$-rule:} &
$\{ \alpha \wedge \beta \} \cupdot \Delta  \;/\;  
( \Delta  \cup \{  \alpha, \beta \} )$.\\
{$U$-rule:} &
$\{  \alpha U \beta  \} \cupdot \Delta  \;/\;  
( \Delta  \cup \{  \beta \} \; |\; 
\Delta  \cup \{ \alpha, X( \alpha U \beta) \}) $.\\
\end{tabular}

There are also static rules
for negations:\\
\begin{tabular}{ll}
{$\neg \neg$-rule:} &
$\{ \neg \neg \alpha \} \cupdot \Delta \;/\; \Delta \cup \{ \alpha \}$.\\
{$\neg \wedge$-rule:} &
$\{ \neg (\alpha \wedge \beta) \} \cupdot \Delta  \;/\;  
( \Delta  \cup \{  \neg \alpha \}
\; |\; \Delta  \cup \{ \neg \beta \} )$.\\
{$\neg U$-rule:} &
$\{  \neg ( \alpha U \beta)  \} \cupdot \Delta  \;/\;  
( \Delta  \cup \{  \neg \alpha, \neg \beta \} \; |\; 
\Delta  \cup \{ \neg \beta, X \neg ( \alpha U \beta) \}) $.\\
\end{tabular}


Although the following rules can be derived from the above,
and although we are trying to abbreviate the explanation here,
the following derived rules may be useful 
for the reader to better understand
the brief examples presented in this paper.
They are static rules for $F$ and $G$ and their
negations:\\
\begin{tabular}{ll}
{$F$-rule:} &
$\{ F \alpha \} \cupdot \Delta  \;/\;  
( \Delta  \cup \{  \alpha \} \; |\; 
\Delta  \cup \{ XF \alpha \}) $.\\
{$G$-rule:} &
$\{ G\alpha \} \cupdot \Delta  \;/\;  
\Delta  \cup \{ \alpha, XG \alpha\}$.\\
{$\neg G$-rule:} &
$\{ \neg G \alpha \} \cupdot \Delta  \;/\;  
( \Delta  \cup \{  \neg \alpha \} \; |\; 
\Delta  \cup \{ X \neg G \alpha \}) $.\\
{$\neg F$-rule:} &
$\{ \neg F \alpha \} \cupdot \Delta  \;/\;  
\Delta  \cup \{ \neg \alpha, X \neg F \alpha\}$.\\
\end{tabular}

The remaining four non-static rules are only applicable
when a label is poised
(which implies that none of the static rules will
be applicable to it).
In presenting them we 
use the convention that
a node $u$ has label $\Gamma_u$.
More than one of the following rules 
may apply to a particular leaf node at any time:
in that case, 
we apply the rule which is listed
here first.

{\bf [LOOP]:}
If a node $v$ 
with poised label $\Gamma_v$
has a 
proper ancestor
(i.e., not itself)
$u$ with poised label $\Gamma_u$ such that
$\Gamma_u \supseteq \Gamma_v$,
and for each $X$-eventuality
$X(\alpha U \beta)$ or $XF \beta$ in $\Gamma_u$
we have 
 a node $w$
such that $u < w \leq v$
and $\beta \in \Gamma_w$
then
$v$ should be a ticked leaf.

{\bf [\prune]:}
Suppose that
$u < v < w$
and each of
$u$, $v$ and $w$
have the same poised label 
$\Gamma$.
Suppose also that
for each $X$-eventuality $X(\alpha U \beta)$
or $XF \beta$ in $ \Gamma$,
if there is $x$ with
$\beta \in \Gamma_x$
and $v < x \leq w$
then
there is $y$ such that
$\beta \in \Gamma_y$
and $u < y \leq v$.
Then $w$ should be a crossed leaf.

{\bf [\prunez]:}
Suppose that
$u < v$
share the same poised label $\Gamma$ and
$\Gamma$ contains at least one $X$-eventuality.
Suppose that
there is 
no $X$-eventuality $X(\alpha U \beta)$ or $XF \beta$ in $\Gamma$
with a node $x$
such that
$\beta \in \Gamma_x$
and $u < x \leq v$.
Then $v$ should be a crossed leaf.

{\bf [TRANSITION]:}
If none of the other rules above do apply to it
then
a node labelled by 
poised $\Gamma$ say,
can have one child
whose label is: 
$\Delta = \{ \alpha |
X \alpha \in \Gamma \}
\cup
\{ \neg \alpha |
\neg X \alpha \in \Gamma \}
$.

The tableau is extended
at each stage by
choosing the leaf node
on any unfinished branch
and attempting to apply a rule:
if no other rules apply, the TRANSITION rule
will.
The choice of unfinished branch is arbitrary.
If there are no unfinished branches,
so the tableau is finished,
then clearly it will either be successful
or failed.


\section{The motivation for the tableaux rules}
\label{sec:motiv}


A traditional classical logic style tableau
starts with the formula
in question and breaks it down into
simpler formulas
as we move down the page.
The simpler formulas
being satisfied should ensure that the
more complicated parent label
is satisfied.
Alternatives are presented as branches.
See the
example given in
Figure~\ref{fig:eg4}.

\begin{figure}
\centering
\begin{minipage}{.5\textwidth}
  \centering
\includegraphics[width=8cm,trim= 5cm 20cm 4cm 4.7cm,clip=true]{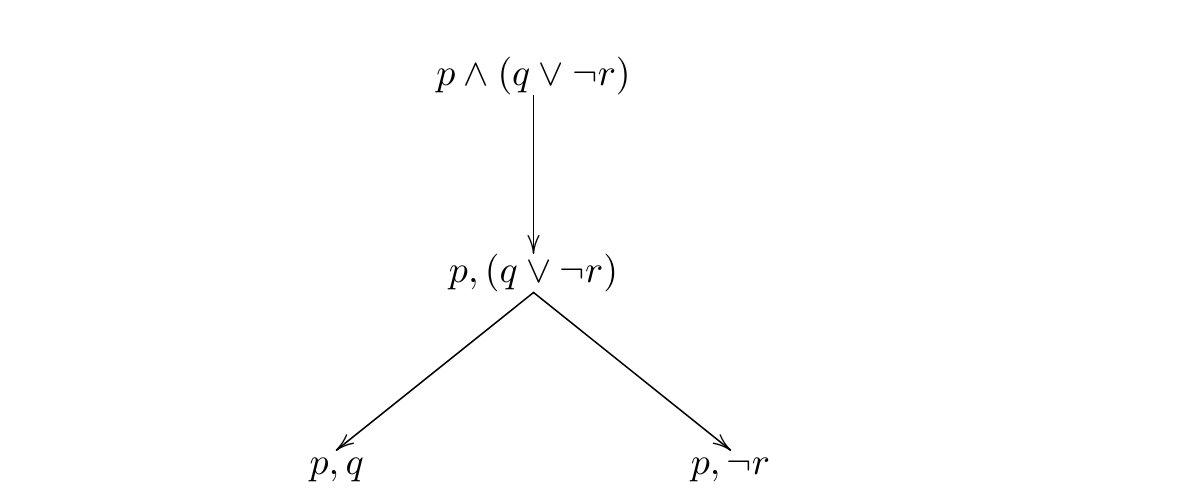}
  \captionof{figure}{Classical disjunction}
  \label{fig:eg4}
\end{minipage}%
\begin{minipage}{.5\textwidth}
  \centering
\includegraphics[width=8cm,trim= 6cm 19cm 4cm 4.7cm,clip=true]{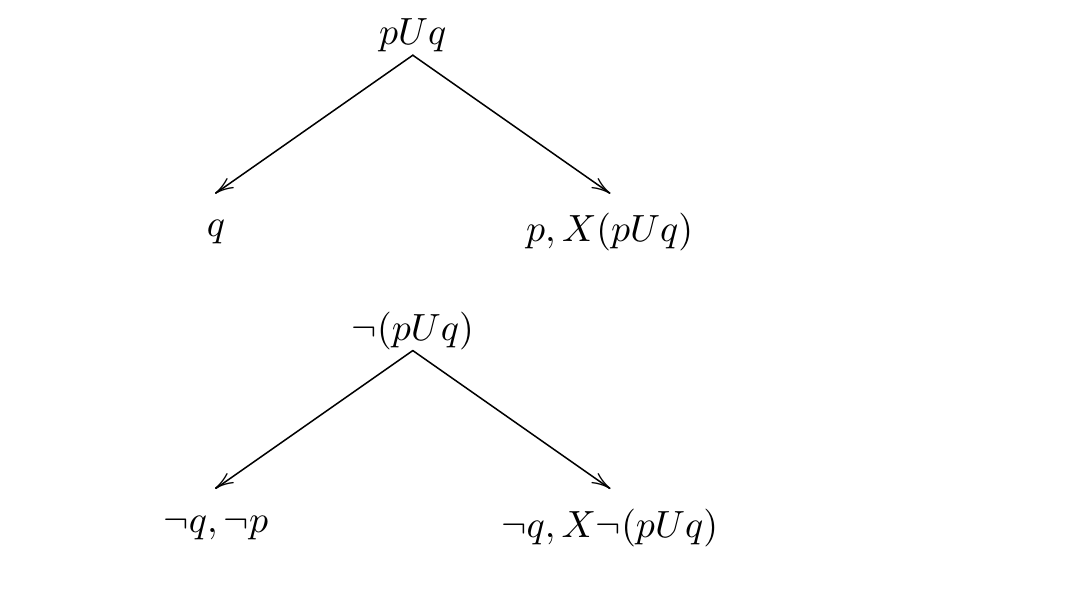}
  \captionof{figure}{Until also gives us choices}
  \label{fig:eg6}
\end{minipage}
\end{figure}

We follow this style of tableau
as is evident by the
classical look of the tableau rules
involving classical connectives.
The $U$ and $\neg U$ rules
are also in this vein,
noting that temporal formulas
such as 
$U$ also gives us choices:
Figure~\ref{fig:eg6}.

Eventually, we break down a formula
into elementary ones.
The atoms and their negations
can be satisfied immediately
provided there are no contradictions,
but to reason about the $X$ formulas
we need to move forwards in time.
\longversion{
How do we do this?
See Figure~\ref{fig:eg7}.

\begin{figure}
\centering
\begin{minipage}{.5\textwidth}
  \centering
\includegraphics[width=8cm,trim= 5cm 18cm 5cm 4.5cm,clip=true]{example7}
  \captionof{figure}{But what to do when we want to move forwards in time?}
  \label{fig:eg7}
\end{minipage}%
\begin{minipage}{.5\textwidth}
  \centering
\includegraphics[width=8cm,trim= 5cm 18cm 5cm 4.5cm,clip=true]{example8}
  \captionof{figure}{Introduce a new type of TRANSITION}
  \label{fig:eg8}
\end{minipage}
\end{figure}

The answer is that we introduce a new type of TRANSITION step:
see
Figure~\ref{fig:eg8}.
Reasoning switches to the next time point
and we carry over only information
nested below $X$ and $\neg X$.

}
\shortversion{
This is where we use the TRANSITION step: see
Figure~\ref{fig:eg9}.
Reasoning switches to the next time point
and we carry over only information
nested below $X$ and $\neg X$.

}

 With just these rules
we can now do the whole
$\neg p \wedge X \neg p \wedge (q Up)$
example. See
Figure~\ref{fig:eg9}.

This example is rather simple, though,
and we need additional rules
to deal with infinite behaviour.
 Consider the example
$Gp$
which,
in the absence of additional rules, gives rise to the very repetitive infinite tableau in
Figure~\ref{fig:eg10}.
\begin{figure}
\centering
\begin{minipage}{.45\textwidth}
  \centering
\includegraphics[width=8cm,trim= 8cm 17cm 4cm 4.7cm,clip=true]{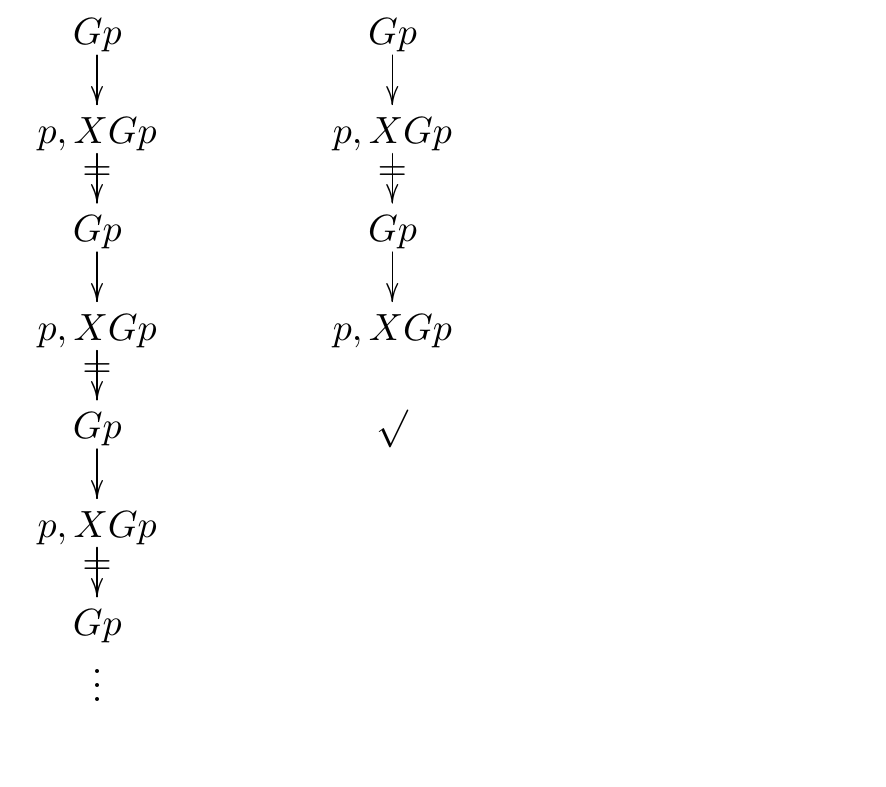}
  \captionof{figure}{$Gp$ gives rise to a very repetitive infinite tableau without the LOOP rule,
  but succeeds quickly with it}
  \label{fig:eg10}
\end{minipage}%
\begin{minipage}{.05\textwidth}
$\;$
\end{minipage}%
\begin{minipage}{.45\textwidth}
  \centering
\includegraphics[width=10cm,trim= 6.5cm 16cm 0.8cm 4.7cm,clip=true]{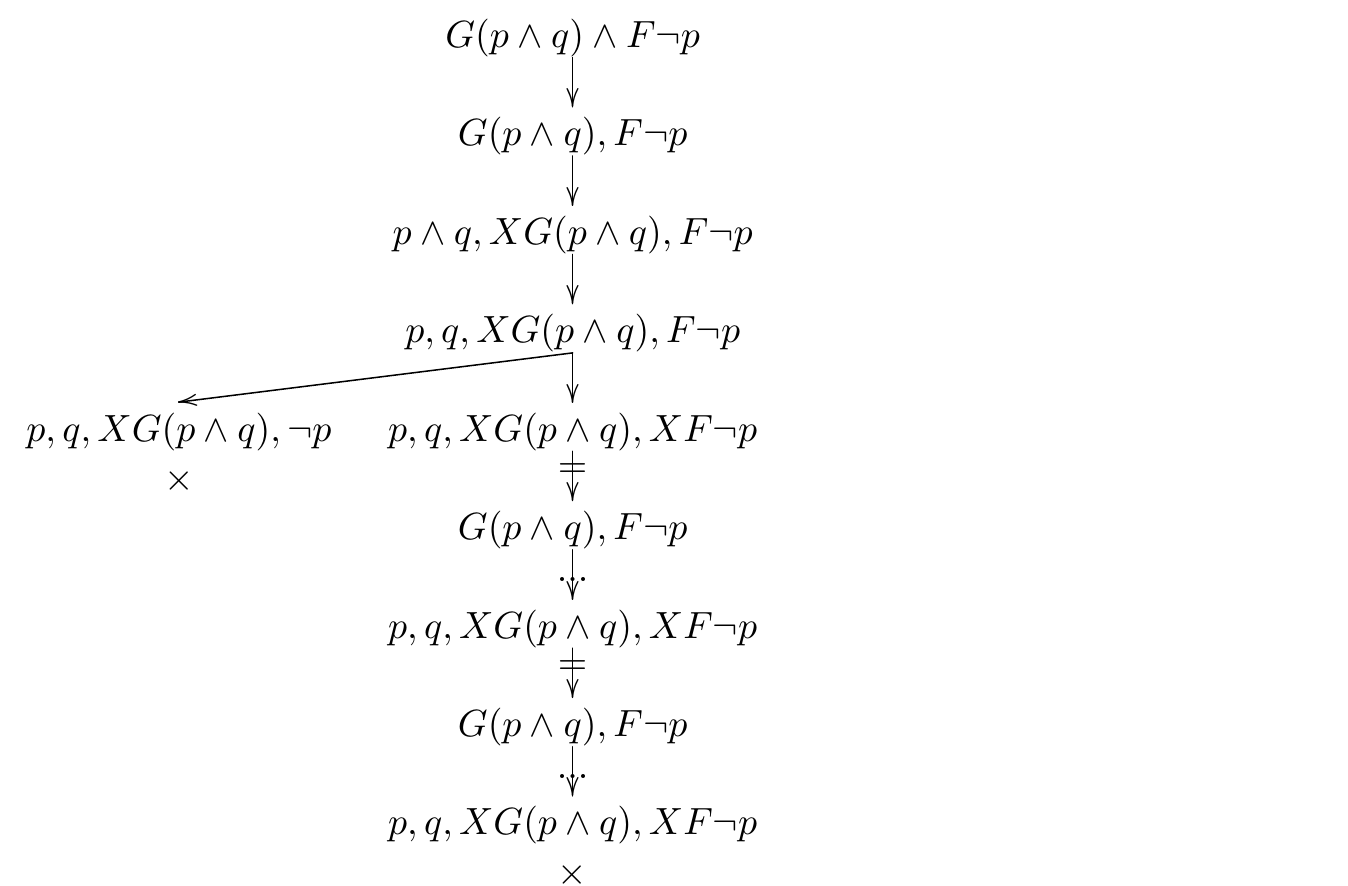}
  \captionof{figure}{$G(p \wedge q) \wedge F \neg p$ crossed by PRUNE}
  \label{fig:eg11}
\end{minipage}
\end{figure}

Notice that the infinite fullpath that it suggests is a model for $Gp$
as would be a fullpath just consisting of the one state
with a self-loop (a transition from itself to itself).

This suggests that we should allow
the tableau branch construction to halt if a
state is repeated.
However the
example $G(p \wedge q) \wedge F \neg p$ shows
that we can not just accept infinite loops
as demonstrating satisfiability:
the tableau for this
unsatisfiable formula would have an 
infinite branch if we did
not use the PRUNE rule
to cross it
(Figure~\ref{fig:eg11}).
Note that the optional \prunez\ rule
can be used to cross the branch
one TRANSITION earlier.

Notice that the infinite fullpath that 
the tableau suggests is this time
not a model for $G(p \wedge q) \wedge F \neg p$.
Constant repeating of
$p, q$ being made true does not
satisfy the conjunct
$F \neg p$.
We have postponed the {\em eventuality}
forever
and this is not acceptable.

If $\alpha U \beta$ appears in 
the tableau label of a node $u$
then we want
$\beta$ to appear in the label of 
some later (or equal node) $v$.
In that case we say that the eventuality
is {\em satisfied}
by $v$.

Eventualities are eventually satisfied
in any (actual) model of a formula:
by the semantics of $U$.

This motivates the LOOP rule.
If a label is repeated along a branch
and all eventualities
are satisfied in between
then
we can build a model
by looping states.
In fact, the ancestor can have a superset
and it will work (see the soundness proof in 
\cite{\techreport}).

Examples like $G(p \wedge q) \wedge F \neg p$
(in Figure~\ref{fig:eg11})
and
$p \wedge G(p \rightarrow Xp)
\wedge F \neg p$
which have branches that go on forever
without satisfying eventualities,
still present a problem for us.
We need to stop and fail branches
so that we can answer ``no'' correctly and terminate
and so that we do not get distracted 
when another branch may be successful.
In fact, no infinite branches should be allowed.

The final rule that we
consider,
and the most novel,
is based on
observing that
these infinite branches
are just getting repetitive without
making a model.
The repetition
is clear because
there are only a finite set of 
formulas
which can ever appear in labels
for a given initial formula $\phi$.
The {\em closure set} for a formula $\phi$
is as follows:
\[
\{
\psi, \neg \psi | \psi \leq \phi \}
\cup 
\{ X (\alpha U \beta),
\neg X (\alpha U \beta) |
\alpha U \beta \leq \phi 
\}
\]
Here we use $\psi \leq \phi$ to mean that
$\psi$ is a subformula
of $\phi$.
The size of closure set is
$\leq 4n$
where $n$ is the length of 
the initial formula.
Only formulas
from this finite set
will appear in labels.
So there are only $\leq 2^{4n}$
possible labels.

The PRUNE rule is as follows.
If a node at the end of a branch
(of an unfinished tableau)
has a label which has 
appeared already twice above,
and
between the second and third
appearance 
there are no new eventualities satisfied
that were not already satisfied between
the first and second appearances
then 
that whole interval of states (second to third appearance)
has been useless.
The \prunez\ rule
applies similar reasoning to an initial
repeat in which no eventualities are fulfilled.
In Figure~\ref{fig:eg11},
the PRUNE rule 
crosses the right hand branch
as the only $X$-eventuality
$XF \neg p$
remains unfulfilled
as $\neg p$ does not appear
in a label despite three repeats of
the same label.

 
It should
be mentioned that the
tableau building process
we describe 
above
is non-deterministic in several respects
and so
really not a proper description of an algorithm.
However, we will see in the
correctness proof below
that the further details 
of which formula to
consider at each step
in building the 
tableau are unimportant.

Finally a suggestion for a nice
example to try.
Try
$p \wedge G(p \leftrightarrow X \neg p)
\wedge G(q \rightarrow \neg p)
\wedge G(r \rightarrow \neg p)
\wedge G(q \rightarrow \neg r)
\wedge GFq
\wedge GFr$.

 \shortversion{
 
 \section{Proof of Correctness: Soundness and Completeness}
\label{sec:correctness}

In the long technical report \cite{\techreport},
we show full details of the proof of soundness, completeness and termination
for the tableau search.
Termination is guaranteed because there can be no infinitely long branches.
Soundness presents no novelty to
those familiar with soundness proofs for similar tableaux:
construct a model from the successful tableau branch.
Completeness, however, is novel, because of the novel rules
and we present that here.


}



Note that the completeness proof we present here seems to be
entirely novel: 
using a model of a satisfiable formula to
find a successful branch in a tableau,
and having to backtrack in general.
Ideas of relating a tableau branch to a model 
appear in many places (eg. \cite{Wid10})
but here we have to backtrack up the tableau
while continuing in the model
and use a subtle progress argument to show that
we can not do this forever.
An interesting aspect here compared to some other
completeness proofs such as that in
\cite{BrL08}
is that we see that with the new PRUNE rule
there is no need to rely on any ``fair"
expansion strategy amongst eventualities:
being unfair for too long falls fowl
of the PRUNE rules.
Note also that there does not seem 
to be a simpler completeness proof based on the idea
of a
shortest lasso model of a formula and
a claim that the PRUNE rule will not apply.

\shortversion{
Note also that the completeness proof
gives us the rather strong result that
for a satisfiable formula 
we will find a successful tableau
regardless of the order in which we use the rules
and regardless of the order in choosing unfinished
branches to extend.
}

\begin{lemma}[Completeness]
Suppose that $\phi$ is a satisfiable formula
of LTL.
Then any finished tableau for $\phi$ will
be successful.
\end{lemma}

\begin{proof}
Suppose that $\phi$ is a satisfiable formula
of LTL.
It will have a model. Choose one, say
$(S,R,g), \sigma \models \phi$.
In what follows we (use standard practice when the
model is fixed and) write
$\sigma_{\geq i} \models \alpha$
when we mean
$(S,R,g), \sigma_{\geq i} \models \alpha$.

Also, build a finished tableau $T$ for
$\phi$ in any manner as long as the rules are followed.
Let $\Gamma(x)$ be the formula
set label on the node $x$ in $T$.
We will show that 
$T$ has a ticked leaf.

To do this we will
construct a sequence
$x_0, x_1, x_2, ....$
of nodes, with $x_0$ being the root.
This sequence may terminate at a tick
(and then we have succeeded)
or it may hypothetically go on forever
(and more on that later).
In general, the sequence 
will head downwards
from a parent to a child node
but 
occasionally it may 
jump back up to an ancestor.

As we go we will
make sure that each
node $x_i$ is associated with
an index $J(i)$ along the fullpath $\sigma$
and we guarantee the following invariant
$INV(x_i,J(i))$ for each $i \geq 0$.
The relationship
$INV(x,j)$ is that
for each $\alpha \in \Gamma(x)$,
$ \sigma_{\geq j} \models \alpha$.

Start by putting $J(0)=0$
when $x_0$ is the tableau root node.
Note that
the only formula in
$\Gamma(x_0)$ is $\phi$ and
that 
$ \sigma_{\geq 0} \models \phi$.
Thus $INV(x_0,J(0))$ holds at the start.

Now suppose that we have identified the $x$ sequence up until
$x_{i}$.
Consider the rule that is used in $T$ to
extend a tableau branch from $x_i$ to some
children.
Note that we can also ignore the 
cases in which the
rule is EMPTY or LOOP
because they would immediately give us the
ticked branch that is sought.

It is useful to define the sequence advancement procedure
in the cases apart from the PRUNE rule separately.
Thus we now describe a procedure,
call it $A$,
that is given a node $x$ and index $j$
satisfying
$INV(x,j)$
and, in case that the node $x$ has children via 
any rule except PRUNE,
the procedure $A$ will give us
a child node $x'$ and index $j'$
which is either $j$ or $j+1$,
such that $INV(x',j')$ holds.
The idea will be to use procedure $A$
on $x_i$ and $J(i)$
to get $x_{i+1}$ and $J(i+1)$
in case the PRUNE rule is 
not used at node $x_i$.
We return to deal with advancing from $x_i$
in case that the PRUNE rule is used later.
So now we describe procedure $A$
with $INV(x,j)$ assumed.


In this conference paper only the
most interesting rules are shown:
see \cite{\techreport} for the rest.

[EMPTY]
If $\Gamma(x)= \{ \}$
then
we are done. $T$ is a successful tableau as required.

[CONTRADICTION]
Consider if it is possible for us to reach
a leaf at $x$ with a cross
because of a contradiction.
So there is some $\alpha$ with
$\alpha$ and $\neg \alpha$ in $\Gamma(x)$.
But this can not happen as then
$ \sigma_{\geq j} \models \alpha$
and
$ \sigma_{\geq j} \models \neg \alpha$.

[$U$-rule]
So $\Gamma(x)= \Delta \cupdot \{ \alpha U \beta \}$
and there are two children.
One $y$ is labelled
$\Gamma(y)= \Delta \cup \{ \beta \}$
 and
the other, $z$,
is labelled
$\Gamma(z)= \Delta \cup \{ \alpha, X(\alpha U \beta) \}$.
We know
$ \sigma_{\geq j} \models \alpha U \beta$.
Thus,
there is some $k \geq j$ such that
$\sigma_{\geq k} \models \beta$
and
for all $l$,
if
$j \leq l < k$
then
$ \sigma_{\geq l} \models \alpha$.
If 
$ \sigma_{\geq j} \models \beta$
then 
we can choose $k=j$
(even if other choices as possible)
and otherwise
choose any such $k > j$.
Again there are two cases,
either $k=j$
or $k> j$.

In the first case, when 
$ \sigma_{\geq j} \models \beta$,
we put
$x'=y$
and otherwise we will
make
$x'=z$.
In either case put
$j'=j$.

Let us check the invariant.
Consider the first case.
We have
$ \sigma_{\geq j'} \models \beta$.

In the second case,
we know that we have
$ \sigma_{\geq j'} \models \alpha$
and
$ \sigma_{\geq j'+1} \models \alpha U \beta$.
Thus
$ \sigma_{\geq j'} \models X( \alpha U \beta)$.

Also, in either case,
for every other $\gamma \in \Gamma(x')$
we still have
$ \sigma_{\geq j'} \models \gamma$.
So we have the invariant holding.

[$\neg U$-rule]
So $\Gamma(x) = \Delta \cupdot \{ \neg( \alpha U \beta)\}$
and there are two children.
One $y$ is labelled
$\Delta \cup \{ \neg \alpha, \neg \beta \}$
 and
the other, $z$,
is labelled
$\Delta \cup \{ \neg \beta, X\neg (\alpha U \beta) \}$.
We know
$ \sigma_{\geq j} \models \neg (\alpha U \beta)$.
So for sure
$ \sigma_{\geq j} \models \neg \beta$.

Furthermore,
possibly
$ \sigma_{\geq j} \models \neg \alpha$
as well,
but otherwise
if
$ \sigma_{\geq j} \models \alpha$
then we can show that we can not have
$ \sigma_{\geq j+1} \models \alpha U \beta$.
Suppose for contradiction
that
$ \sigma_{\geq j} \models \alpha$
and
$\sigma_{\geq j+1} \models \alpha U \beta$.
Then 
there is some $k \geq j$ such that
$ \sigma_{\geq k} \models \beta$
and
for all $l$,
if
$j \leq l < k$
then
$ \sigma_{\geq l} \models \alpha$.
Thus
$\sigma_{\geq j} \models \alpha U \beta$.
Contradiction.

So we can conclude that there
are two cases when the $\neg U$-rule is used.
CASE 1:
$ \sigma_{\geq j} \models \neg \beta$
and
$ \sigma_{\geq j} \models \neg \alpha$.
CASE 2:
$ \sigma_{\geq j} \models \neg \beta$
and
$ \sigma_{\geq j+1} \models \neg (\alpha U \beta)$.

In the first case, when 
$\sigma_{\geq j} \models \neg \beta$,
we put
$x'=y$
and otherwise we will
make
$x'=z$.
In either case put
$j'=j$.

Let us check the invariant.
In both cases
we know that we have
$ \sigma_{\geq j'} \models \neg \beta$.
Now consider the first case.
We also have 
$ \sigma_{\geq j} \models \neg \alpha$.
In the second case,
we know that we have
$ \sigma_{\geq j+1} \models \neg (\alpha U \beta)$.
Thus
$ \sigma_{\geq j'} \models X\neg (\alpha U \beta)$.
Also, in either case,
for every other $\gamma \in \Gamma(x')$
we still have
$ \sigma_{\geq j'} \models \gamma$.
So we have the invariant holding.

[OTHER STATIC RULES]: similar.

{
[TRANSITION]}
So $\Gamma(x)$ is poised
and there is one child,
which we will make
$x'$
and we will put
$j'=j+1$.

Consider a formula
 $\gamma \in \Gamma(x')=
\{ \alpha |
X \alpha \in \Gamma(x) \}
\cup
\{ \neg \alpha |
\neg X \alpha \in \Gamma(x) \}$.

CASE 1: Say that
$X \gamma \in \Gamma(x)$.
Thus, by the invariant,
$\sigma_{\geq j} \models X \gamma$.
Hence,
$\sigma_{\geq j+1} \models \gamma$.
But this is just
$\sigma_{\geq j'} \models \gamma$
as required.

CASE 2: Say that $\gamma= \neg \delta$
and
$\neg X \delta \in \Gamma(x)$.
Thus, by the invariant,
$\sigma_{\geq j} \models \neg X \delta$.
Hence,
$\sigma_{\geq j+1} \not \models \delta$.
But this is just
$\sigma_{\geq j'} \models \gamma$
as required.

So we have the invariant holding.

{
[LOOP]}
If, in $T$, the node $x_i$
is a leaf just getting a tick via the 
LOOP rule
then
we are done. 
$T$ is a successful tableau as required.

So that ends the description of procedure $A$
that is given a node $x$ and index $j$
satisfying
$INV(x,j)$
and, in case that the node $x$ has children via 
any rule except PRUNE (or \prunez)
the procedure $A$ will give us
a child node $x'$ and index $j'$,
which is either $j$ or $j+1$,
such that $INV(x',j')$ holds.
We use procedure $A$ to
construct a sequence
$x_0, x_1, x_2, ....$
of nodes, with $x_0$ being the root.
and guarantee the invariant
$INV(x_i,J(i))$ for each $i \geq 0$.

The idea will be to use procedure $A$
on $x_i$ and $J(i)$
to get $x_{i+1}$ and $J(i+1)$
in case the PRUNE rule is 
not used at node $x_i$.
Start by putting $J(0)=0$
when $x_0$ is the tableau root node.
We have seen that $INV(x_0,J(0))$ holds at the start.

In the conference paper version of the 
rest of the proof we ignore the
\prunez\ rule.
It can be dealt with in a similar way.

{[\prune\ ]}
Now, we complete the description of the construction
of the $x_i$ sequence by 
explaining what to do 
in case
$x_i$ is a node on which PRUNE
is used.
Suppose that
$x_i$
is a node which gets a cross
in $T$
via the \prune\ rule.
So there is a sequence
$u=x_h, x_{h+1}, ...,
x_{h+a}=v, x_{h+a+1}, ...,
x_{h+a+b}=x_i=w$
such that
$\Gamma(u)=\Gamma(v)=\Gamma(w)$
and no extra eventualities of
$\Gamma(u)$ are satisfied between
$v$ and $w$ that were
not already satisfied
between $u$ and $v$.

What we do now is to 
undertake a sort of backtracking exercise
in our proof.
We
choose some such
$u$, $v$ and $w$,
there may be more than one triple,
and
proceed with the construction
as if 
$x_i$ was $v$ instead of $w$.
That is we use
the procedure
$A$
on $v$ with $J(i)$
 to get from $v$ to
 one $x_{i+1}$ of its children
 and define $J(i+1)$.
Procedure $A$ above 
can be applied because
$\Gamma(v)=\Gamma(x_i)$
and so 
$INV(v,J(i))$ holds.
The procedure $A$ gives us 
$INV(x_{i+1},J(i+1))$ as well.

If we find $x_{i+1}$ from $x_i$ 
in this way when $x_i$ is a 
tableau node on which PRUNE is applied
then
we say that
$x_{i+1}$ is obtained from $x_i$
via the {\em jump tuple} $(u,v,w)$.

Thus we keep going
with the new $x_{i+1}$ child of $v$,
and $J(i)$.
We have made the subtle move of
backtracking up the tableau to find our
$x_{i+1}$ but
simultaneously continued progress
to $J(i+1)$
without backtracking 
in our journey through the model.

That ends our complete description
of how to find the matching
sequences of $x_i$ and $J(i)$.
We have seen that 
the sequences either finish at a ticked node of the tableau
or can go on another step.

To finish the proof
we need to consider 
whether the above construction
goes on for ever.
Clearly it may
end finitely with
us 
finding a ticked leaf and
succeeding.
However,
at least in theory,
it may seem possible that the
construction keeps going forever
even though the tableau will be 
finite.
The rest of the proof
is to show that this 
actually can not happen.
The construction can not go on forever.
It must stop and
the only way that we have 
shown that that can happen is
by finding a tick.

Suppose for contradiction that the construction does go on forever.
Thus, because there are only a finite number of nodes in the tableau,
and because procedure $A$
defines $x_{i+1}$ as a child of $x_i$,
then
we must meet the \prune\ 
rule and
jump back up the tableau infinitely often.

\newcommand{\triple}{{tuple}}

There are only a finite number
of nodes in $T$ so
only a finite number of jump \triple s
so there must be some that
are used to obtain $x_{i+1}$
for infinitely many $i$.
Call these {\em recurring}
jump \triple s.

Say that
$(u_0,v_0, w_0)$ 
is
one such recurring tuple.
Choose $u_0$ so that 
for no other recurring jump triple
$(u_1,  v_1, w_1)$
do we have
$u_1$ being a proper ancestor of $u_0$.

As we proceed through the construction
of $x_0, x_1, ..$
and see a jump every so often,
eventually all the 
jump \triple s
which only cause a jump a finite number
of times
stop causing jumps.
After that index, $Z$ say,
$(u_0,v_0, w_0)$
will still cause a jump
every so often.

Thus
$u_0$ will never appear
again as $x_i$ 
for $i >Z$
and all $x_i$ for $i>Z$ that we choose will
be descendants of $u_0$.
This is because by choice of $u_0$ we will never jump
up to $u_0$ or above it (closer to the root)
via any jump tuple that is used after $Z$.
Say that $x_N$ is the very last $x_i$ that
is equal to $u_0$.

Now consider any $X(\alpha U \beta)$ that appears
in $\Gamma(u_0)$.
(There must be at least one eventuality
in $\Gamma(u_0)$ as it is used to
apply rule \prune\
and not the LOOP rule).

A simple induction shows that
$\alpha U \beta$
or $X( \alpha U \beta)$
will appear in every
$\Gamma(x_i)$
from $i=N$
up until at least when
$\beta$ appears
in some $\Gamma(x_i)$ for $i>N$
(if that ever happens).
This is because
if $\alpha U \beta$
is in $\Gamma(x_i)$
then it will also be in any child node
unless the UNTIL rule is used.
If the UNTIL rule is used on $x_i$
and $\beta$ is not in $\Gamma(x_i)$
and does not get put in $\Gamma(x_{i+1})$
then
$X(\alpha U \beta)$ will be put
in $\Gamma(x_{i+1})$.
The subsequent temporal TRANSITION rule
will thus put $\alpha U \beta$
into the new label.
Finally, in case the $x_i$ sequence meets a PRUNE jump
$(u,v,w)$
then the new $x_{i+1}$ will be a child 
of $v$ which is a descendent of $u$
which is a descendent of $u_0$
so will also contain 
$\alpha U \beta$
or $X(\alpha U \beta)$.

Now $J(i)$ just
increases by $0$ or $1$ 
with each increment of $i$,
We also know that
$\sigma_{\geq J(i)} \models \alpha U \beta$
from $i=N$ onwards
until (and if) $\beta$ gets put in $\Gamma(x_i)$.
Since $\sigma$ is a fullpath
we will eventually get to some
$i$
with
$\sigma_{\geq J(i)} \models \beta$.
In that case our
construction makes us put
$\beta$ in the label (to keep the invariant holding).
Thus we do eventually get to some
$i \geq N$
with $\beta \in \Gamma(x_i)$.
Let $N_\beta$ be the first
such $i \geq N$.
Note that
all the nodes between
$u_0$ and
$x_{N_\beta}$
in the tableau
also appear as
$x_i$ for 
$N < i < N_\beta$
so that
they all have
$\alpha U \beta$ and not $\beta$
in their labels
$\Gamma(x_i)$.

Now let us consider if 
we ever
jump up above 
$x_{N_\beta}$
at any TRANSITION of our construction
(after $i=N_\beta$).
In that case 
there would be a PRUNE jump triple
of tableau nodes
$u$, $v$ and $w$
governing the first such jump.
Since $u$ is not above
$u_0$
and $v$ is above 
$x_{N_\beta}$,
we must have 
$\Gamma(u) = \Gamma(v)$
with 
$X(\alpha U \beta)$
in them
and $\beta$ not satisfied in between.
But $w$ will be below
$x_{N_\beta}$
at the first such jump,
meaning that
$\beta$ is satisfied
between
$v$ and $w$.
That is
a contradiction to the PRUNE rule being applicable to this triple.

Thus the $x_i$ sequence 
stays within descendants of $x_{N_\beta}$
forever after $N_\beta$.

The above reasoning applies to all
eventualities in $\Gamma(u_0)$.
Thus, after they are each satisfied,
the construction $x_i$
does not jump up above any of them.
When the next supposed
jump involving
$u_0$ with some $v$ and
 $w$
happens after that
it is clear that
all of the eventualities
in $\Gamma(u_0)$
are satisfied above $v$.
Thus the LOOP rule would have applied 
between $u$ and $v$.

This is a contradiction
to such a jump ever happening.
Thus we can conclude that
there are not an infinite number of jumps after all.
The construction must finish with a tick.
This is the end of the completeness proof.
\end{proof}

\section{Implementations}
\label{sec:complex}



Tableau search here,
even in a non-parallel implementation,
should (theoretically) be able to be implemented to run
a little faster than that 
the state of the art tableau technique
of \cite{Schwe98}.
This is because
there is less information to keep track of and
no backtracking from potentially successful branches
when a repeated label is discovered.

A fast implementation of the new tableau
written by
Matteo Battelo
of Udine University
is available from 
\url{https://github.com/Corralx/leviathan}
and described in
\cite{\ijcai}.
Experiments run using this implementation,
on the full set of 3723 standard benchmark formulas
\cite{VSchuppanLDarmawan-ATVA-2011},
show comparative speed performance with
five state of the art tools 
(Aalta \cite{LiY2014}, 
TRP++ \cite{hustadt2003trp++}, 
LS4 \cite{SuW2012}, 
NuSMV \cite{Cimatti2002}, 
and 
PLTL) based
on automata, resolution, resolution(again), symbolic model checking,
the Schwendimann tableau technique respectively. 
Interestingly the memory usage for the
new tableau is significantly less.
See \cite{\ijcai}
for the details
of the implementation and experiments.

For now, 
as this current paper is primarily about the
theory behind the new rules,
we have provided
a demonstration Java implementation
to allow readers to experiment with
the way that the tableau works.
The program allows comparison with 
a corresponding implementation of
the Schwendimann
tableau.
The demonstration Java implementation
is available at
\webpage.
This allows users to understand the tableau building process
in a step by step way.
It is not 
designed as a fast implementation.
However, it does report on
how many tableau construction steps were taken.

Detailed comparisons of the running times
across 24 typical benchmark formulas
are available in \cite{\techreport}.
In Figure~\ref{fig:compare},
we give a small selection to give the
idea of the experimental results.
This is just on two quite long formulas,
``Rozier 6'' and ``Rozier 9"
and one very long formula
``anzu amba amba 6"
from the so-called Rozier counter example series
of \cite{VSchuppanLDarmawan-ATVA-2011}.
Shown is formula length,
running time in seconds (on a standard laptop),
number of tableau steps and
the maximum depth of a branch in
poised states.
As claimed,
the new tableau needs roughly the same number
of steps but saves time
on each step (at least as the formulas get longer).
Indeed there are only three formulas
presented but they are each part of important
series which each show the 
same pattern:
the two approaches are comparable in
terms of number of steps taken.
Future heavy duty implementation
should then be able to deliver
on a faster implementation for the new approach as
we know each step is simpler than
a step for the Schwendimann approach.

\begin{figure}
\begin{center}
$\;$\\
\begin{tabular}{|c|c|c|c|c|c|c|c|}
\hline
\multicolumn{2}{|c|}{fmla} &  \multicolumn{3}{c|}{Reynolds} & \multicolumn{3}{c|}{Schwendimann}\\
\hline
& length & sec & steps & depth & sec & steps & depth \\
\hline
r6 & 190 & 0.871 & 20k & 353 & 0.855 & 20k & 353 \\
r9 & 277 & 113 & 239k & 4353 & 120 & 242k & 4353 \\
as6 & 1864 & 0.003 & 54k & 2 & 0.015 & 61k & 2 \\
\hline
\end{tabular}
\end{center}
\caption{Comparison of the two tableaux from the Java implementation}
\label{fig:compare}
\end{figure}

\section{Example}
\label{sec:example}



In order to illustrate the differences in approach
very briefly we include two diagrams from the
long paper.
See Figure~\ref{fig:eg18} and  \ref{fig:eg19} for the overall shape
of the two tableaux applied to 
the example
\[
\theta= 
a
\wedge 
G(a \leftrightarrow X \neg a)
\wedge GF b_1 
\wedge GF b_2 
\wedge G( b_1  \rightarrow  \neg a)
\wedge G( b_2    \rightarrow   \neg a)
\wedge G \neg( b_1 \wedge b_2 ).
\]
Roughly, there are two eventualities which need respectively state $1$ and $3$
to fulfil them but we must return to state $0$ in between.
There are some similar observations
on the Schwendiman approach in
\cite{How97}.

\begin{figure}
\centering
\begin{minipage}{.5\textwidth}
  \centering
\includegraphics[width=4cm,trim= 4.5cm 12cm 4.5cm 10cm,clip=true]{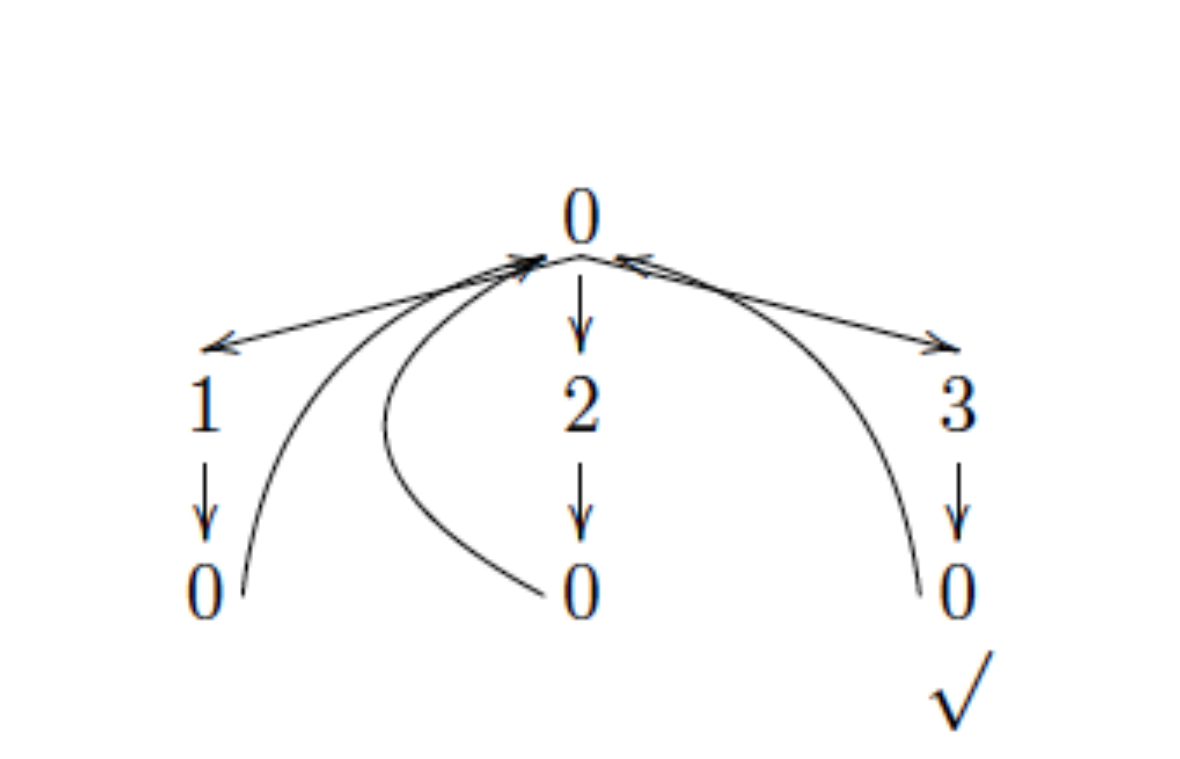}
  \captionof{figure}{Schwendimann Example}
  \label{fig:eg18}
\end{minipage}%
\begin{minipage}{.5\textwidth}
  \centering
\includegraphics[width=4cm,trim= 4.5cm 10cm 4.5cm 10cm,clip=true]{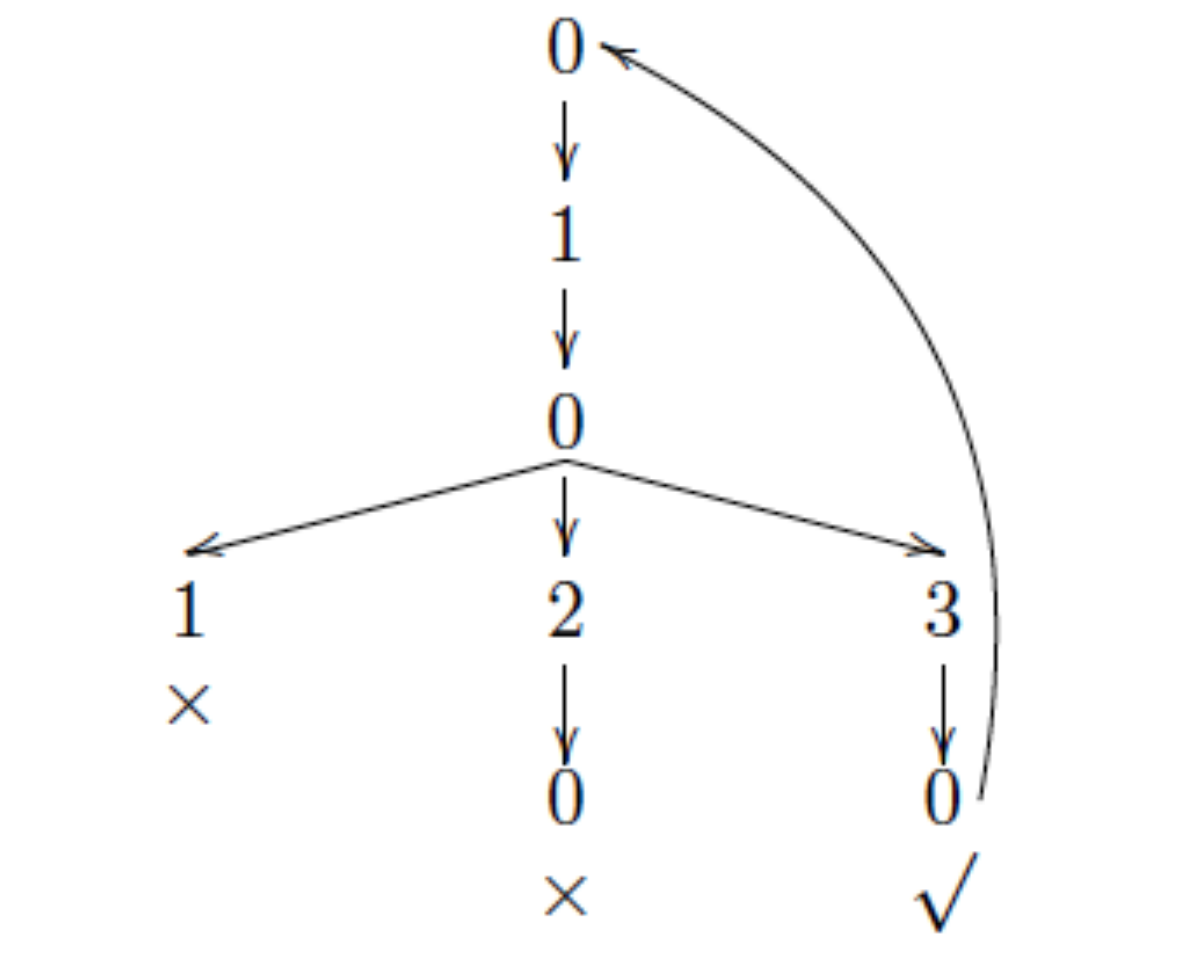}
  \captionof{figure}{Same Example with new tableau}
  \label{fig:eg19}
\end{minipage}
\end{figure}


\section{Conclusion}
\label{sec:concl}

We have introduced a new
type of tableau rule for temporal logics,
in particular for LTL.
This allows the operation of a novel but traditionally
tree-shaped, one-pass tableau system
for LTLSAT.
It is simple in all aspects
with no extra notations on nodes,
neat to introduce to students,
amenable to manual use and
can be implemented
efficiently 
with competitive performance.

In searching or constructing the tableau
one can explore down branches completely independently
and further break up the search down individual
branches into separate somewhat independent
processes.
Thus it is particularly suited to parallel implementations.

Because of the simplicity,
it also seems to be a good base for more intelligent
and sophisticated algorithms:
including heuristics for choosing amongst
branches and ways of managing sequences
of label sets.
The idea of the PRUNE rules
potentially have many other applications.

\bibliographystyle{eptcs}

\begin{thebibliography}{10}
\providecommand{\bibitemdeclare}[2]{}
\providecommand{\surnamestart}{}
\providecommand{\surnameend}{}
\providecommand{\urlprefix}{Available at }
\providecommand{\url}[1]{\texttt{#1}}
\providecommand{\href}[2]{\texttt{#2}}
\providecommand{\urlalt}[2]{\href{#1}{#2}}
\providecommand{\doi}[1]{doi:\urlalt{http://dx.doi.org/#1}{#1}}
\providecommand{\bibinfo}[2]{#2}

\bibitemdeclare{incollection}{Mor12}
\bibitem{Mor12}
\bibinfo{author}{Mordechai \surnamestart Ben-Ari\surnameend}
  (\bibinfo{year}{2012}): \emph{\bibinfo{title}{Propositional Logic: Formulas,
  Models, Tableaux}}.
\newblock In: {\sl \bibinfo{booktitle}{Mathematical Logic for Computer
  Science}}, \bibinfo{publisher}{Springer London}, pp. \bibinfo{pages}{7--47},
  \doi{10.1007/978-1-4471-4129-7_2}.


\bibitemdeclare{article}{BMP83}
\bibitem{BMP83}
\bibinfo{author}{Mordechai \surnamestart Ben{-}Ari\surnameend},
  \bibinfo{author}{Amir \surnamestart Pnueli\surnameend} \&
  \bibinfo{author}{Zohar \surnamestart Manna\surnameend}
  (\bibinfo{year}{1983}): \emph{\bibinfo{title}{The Temporal Logic of Branching
  Time}}.
\newblock {\sl \bibinfo{journal}{Acta Inf.}} \bibinfo{volume}{20}, pp.
  \bibinfo{pages}{207--226}, 
  \doi{10.1007/BF01257083}.


\bibitemdeclare{inproceedings}{DBLP:conf/ijcai/BertelloGMR16}
\bibitem{DBLP:conf/ijcai/BertelloGMR16}
\bibinfo{author}{Matteo \surnamestart Bertello\surnameend},
  \bibinfo{author}{Nicola \surnamestart Gigante\surnameend},
  \bibinfo{author}{Angelo \surnamestart Montanari\surnameend} \&
  \bibinfo{author}{Mark \surnamestart Reynolds\surnameend}
  (\bibinfo{year}{2016}): \emph{\bibinfo{title}{Leviathan: {A} New {LTL}
  Satisfiability Checking Tool Based on a One-Pass Tree-Shaped Tableau}}.
\newblock In \bibinfo{editor}{Subbarao \surnamestart Kambhampati\surnameend},
  editor: {\sl \bibinfo{booktitle}{Proceedings of the Twenty-Fifth
  International Joint Conference on Artificial Intelligence, {IJCAI} 2016, New
  York, NY, USA, 9-15 July 2016}}, \bibinfo{publisher}{{IJCAI/AAAI} Press}, pp.
  \bibinfo{pages}{950--956}.
\newblock \urlprefix\url{http://www.ijcai.org/Abstract/16/139}.


\bibitemdeclare{article}{Beth55}
\bibitem{Beth55}
\bibinfo{author}{E.~\surnamestart Beth\surnameend} (\bibinfo{year}{1955}):
  \emph{\bibinfo{title}{Semantic Entailment and Formal Derivability}}.
\newblock {\sl \bibinfo{journal}{Mededelingen der Koninklijke Nederlandse Akad.
  van Wetensch}} \bibinfo{volume}{18}.

\bibitemdeclare{inproceedings}{BEM96}
\bibitem{BEM96}
\bibinfo{author}{Julian~C. \surnamestart Bradfield\surnameend},
  \bibinfo{author}{Javier \surnamestart Esparza\surnameend} \&
  \bibinfo{author}{Angelika \surnamestart Mader\surnameend}
  (\bibinfo{year}{1996}): \emph{\bibinfo{title}{An Effective Tableau System for
  the Linear Time {\(\mathrm{\mu}\)}-Calculus}}.
\newblock In \bibinfo{editor}{Friedhelm~Meyer \surnamestart auf~der
  Heide\surnameend} \& \bibinfo{editor}{Burkhard \surnamestart
  Monien\surnameend}, editors: {\sl \bibinfo{booktitle}{Automata, Languages and
  Programming, 23rd International Colloquium, ICALP96, Paderborn, Germany, 8-12
  July 1996, Proceedings}}, {\sl \bibinfo{series}{Lecture Notes in Computer
  Science}} \bibinfo{volume}{1099}, \bibinfo{publisher}{Springer}, pp.
  \bibinfo{pages}{98--109}, 
  \doi{10.1007/3-540-61440-0_120}.


\bibitemdeclare{article}{BrL08}
\bibitem{BrL08}
\bibinfo{author}{Kai \surnamestart Br{\"{u}}nnler\surnameend} \&
  \bibinfo{author}{Martin \surnamestart Lange\surnameend}
  (\bibinfo{year}{2008}): \emph{\bibinfo{title}{Cut-free sequent systems for
  temporal logic}}.
\newblock {\sl \bibinfo{journal}{J. Log. Algebr. Program.}}
  \bibinfo{volume}{76}(\bibinfo{number}{2}), pp. \bibinfo{pages}{216--225},
  \doi{10.1016/j.jlap.2008.02.004}.


\bibitemdeclare{inproceedings}{Cimatti2002}
\bibitem{Cimatti2002}
\bibinfo{author}{Alessandro \surnamestart Cimatti\surnameend},
  \bibinfo{author}{Edmund~M. \surnamestart Clarke\surnameend},
  \bibinfo{author}{Enrico \surnamestart Giunchiglia\surnameend},
  \bibinfo{author}{Fausto \surnamestart Giunchiglia\surnameend},
  \bibinfo{author}{Marco \surnamestart Pistore\surnameend},
  \bibinfo{author}{Marco \surnamestart Roveri\surnameend},
  \bibinfo{author}{Roberto \surnamestart Sebastiani\surnameend} \&
  \bibinfo{author}{Armando \surnamestart Tacchella\surnameend}
  (\bibinfo{year}{2002}): \emph{\bibinfo{title}{NuSMV 2: An OpenSource Tool for
  Symbolic Model Checking}}.
\newblock In \bibinfo{editor}{Ed~\surnamestart Brinksma\surnameend} \&
  \bibinfo{editor}{Kim~Guldstrand \surnamestart Larsen\surnameend}, editors:
  {\sl \bibinfo{booktitle}{Computer Aided Verification, 14th International
  Conference, {CAV} 2002,Copenhagen, Denmark, July 27-31, 2002, Proceedings}},
  {\sl \bibinfo{series}{Lecture Notes in Computer Science}}
  \bibinfo{volume}{2404}, \bibinfo{publisher}{Springer}, pp.
  \bibinfo{pages}{359--364}, \doi{10.1007/3-540-45657-0_29}.


\bibitemdeclare{article}{CGH97}
\bibitem{CGH97}
\bibinfo{author}{Edmund~M. \surnamestart Clarke\surnameend},
  \bibinfo{author}{Orna \surnamestart Grumberg\surnameend} \&
  \bibinfo{author}{Kiyoharu \surnamestart Hamaguchi\surnameend}
  (\bibinfo{year}{1997}): \emph{\bibinfo{title}{Another Look at {LTL} Model
  Checking}}.
\newblock {\sl \bibinfo{journal}{Formal Methods in System Design}}
  \bibinfo{volume}{10}(\bibinfo{number}{1}), pp. \bibinfo{pages}{47--71},
  \doi{10.1023/A:1008615614281}.


\bibitemdeclare{book}{Fos95}
\bibitem{Fos95}
\bibinfo{author}{Ian~T. \surnamestart Foster\surnameend}
  (\bibinfo{year}{1995}): \emph{\bibinfo{title}{Designing and building parallel
  programs - concepts and tools for parallel software engineering}}.
\newblock \bibinfo{publisher}{Addison-Wesley}.

\bibitemdeclare{article}{DBLP:journals/entcs/GaintzarainHLN08}
\bibitem{DBLP:journals/entcs/GaintzarainHLN08}
\bibinfo{author}{Joxe \surnamestart Gaintzarain\surnameend},
  \bibinfo{author}{Montserrat \surnamestart Hermo\surnameend},
  \bibinfo{author}{Paqui \surnamestart Lucio\surnameend} \&
  \bibinfo{author}{Marisa \surnamestart Navarro\surnameend}
  (\bibinfo{year}{2008}): \emph{\bibinfo{title}{Systematic Semantic Tableaux
  for {PLTL}}}.
\newblock {\sl \bibinfo{journal}{Electr. Notes Theor. Comput. Sci.}}
  \bibinfo{volume}{206}, pp. \bibinfo{pages}{59--73},
  \doi{10.1016/j.entcs.2008.03.075}.


\bibitemdeclare{book}{Girle00}
\bibitem{Girle00}
\bibinfo{author}{Rod \surnamestart Girle\surnameend} (\bibinfo{year}{2000}):
  \emph{\bibinfo{title}{Modal Logics and Philosophy}}.
\newblock \bibinfo{publisher}{Acumen, Teddington, UK}.

\bibitemdeclare{article}{Goranko2010113}
\bibitem{Goranko2010113}
\bibinfo{author}{Valentin \surnamestart Goranko\surnameend},
  \bibinfo{author}{Angelo \surnamestart Kyrilov\surnameend} \&
  \bibinfo{author}{Dmitry \surnamestart Shkatov\surnameend}
  (\bibinfo{year}{2010}): \emph{\bibinfo{title}{Tableau Tool for Testing
  Satisfiability in LTL: Implementation and Experimental Analysis}}.
\newblock {\sl \bibinfo{journal}{Electronic Notes in Theoretical Computer
  Science}} \bibinfo{volume}{262}(\bibinfo{number}{0}), pp. \bibinfo{pages}{113
  -- 125}, \doi{10.1016/j.entcs.2010.04.009}.
\newblock \bibinfo{note}{Proceedings of the 6th Workshop on Methods for
  Modalities (M4M-6 2009)}.

\bibitemdeclare{techreport}{Gou89}
\bibitem{Gou89}
\bibinfo{author}{G.~\surnamestart Gough\surnameend} (\bibinfo{year}{1989}):
  \emph{\bibinfo{title}{Decision procedures for temporal logics}}.
\newblock \bibinfo{type}{Technical Report} \bibinfo{number}{UMCS-89-10-1},
  \bibinfo{institution}{Department of Computer Science, University of
  Manchester}.

\bibitemdeclare{inproceedings}{HSZ96}
\bibitem{HSZ96}
\bibinfo{author}{Alain \surnamestart Heuerding\surnameend},
  \bibinfo{author}{Michael \surnamestart Seyfried\surnameend} \&
  \bibinfo{author}{Heinrich \surnamestart Zimmermann\surnameend}
  (\bibinfo{year}{1996}): \emph{\bibinfo{title}{Efficient Loop-Check for
  Backward Proof Search in Some Non-classical Propositional Logics}}.
\newblock In \bibinfo{editor}{Pierangelo \surnamestart Miglioli\surnameend},
  \bibinfo{editor}{Ugo \surnamestart Moscato\surnameend},
  \bibinfo{editor}{Daniele \surnamestart Mundici\surnameend} \&
  \bibinfo{editor}{Mario \surnamestart Ornaghi\surnameend}, editors: {\sl
  \bibinfo{booktitle}{TABLEAUX}}, {\sl \bibinfo{series}{Lecture Notes in
  Computer Science}} \bibinfo{volume}{1071}, \bibinfo{publisher}{Springer}, pp.
  \bibinfo{pages}{210--225}.
  \doi{10.1007/3-540-61208-4_14}.

\bibitemdeclare{inproceedings}{How97}
\bibitem{How97}
\bibinfo{author}{Jacob~M. \surnamestart Howe\surnameend}
  (\bibinfo{year}{1997}): \emph{\bibinfo{title}{Two Loop Detection Mechanisms:
  {A} Comparision}}.
\newblock In \bibinfo{editor}{Didier \surnamestart Galmiche\surnameend},
  editor: {\sl \bibinfo{booktitle}{Automated Reasoning with Analytic Tableaux
  and Related Methods, International Conference, {TABLEAUX} '97,
  Pont-{\`{a}}-Mousson, France, May 13-16, 1997, Proceedings}}, {\sl
  \bibinfo{series}{Lecture Notes in Computer Science}} \bibinfo{volume}{1227},
  \bibinfo{publisher}{Springer}, pp. \bibinfo{pages}{188--200},
  \doi{10.1007/BFb0027414}.


\bibitemdeclare{incollection}{hustadt2003trp++}
\bibitem{hustadt2003trp++}
\bibinfo{author}{Ullrich \surnamestart Hustadt\surnameend} \&
  \bibinfo{author}{Boris \surnamestart Konev\surnameend}
  (\bibinfo{year}{2003}): \emph{\bibinfo{title}{{TRP++2.0:} {A} Temporal
  Resolution Prover}}.
\newblock In \bibinfo{editor}{Franz \surnamestart Baader\surnameend}, editor:
  {\sl \bibinfo{booktitle}{Automated Deduction - CADE-19, 19th International
  Conference on Automated Deduction Miami Beach, FL, USA, July 28 - August 2,
  2003, Proceedings}}, {\sl \bibinfo{series}{Lecture Notes in Computer
  Science}} \bibinfo{volume}{2741}, \bibinfo{publisher}{Springer}, pp.
  \bibinfo{pages}{274--278}, \doi{10.1007/978-3-540-45085-6_21}.


\bibitemdeclare{inproceedings}{DBLP:conf/cav/KestenMMP93}
\bibitem{DBLP:conf/cav/KestenMMP93}
\bibinfo{author}{Yonit \surnamestart Kesten\surnameend}, \bibinfo{author}{Zohar
  \surnamestart Manna\surnameend}, \bibinfo{author}{Hugh \surnamestart
  McGuire\surnameend} \& \bibinfo{author}{Amir \surnamestart Pnueli\surnameend}
  (\bibinfo{year}{1993}): \emph{\bibinfo{title}{A Decision Algorithm for Full
  Propositional Temporal Logic}}.
\newblock In \bibinfo{editor}{Costas \surnamestart Courcoubetis\surnameend},
  editor: {\sl \bibinfo{booktitle}{CAV}}, {\sl \bibinfo{series}{Lecture Notes
  in Computer Science}} \bibinfo{volume}{697}, \bibinfo{publisher}{Springer},
  pp. \bibinfo{pages}{97--109}.
  \doi{10.1007/3-540-56922-7_9}.

\bibitemdeclare{inproceedings}{LiY2014}
\bibitem{LiY2014}
\bibinfo{author}{Jianwen \surnamestart Li\surnameend}, \bibinfo{author}{Yinbo
  \surnamestart Yao\surnameend}, \bibinfo{author}{Geguang \surnamestart
  Pu\surnameend}, \bibinfo{author}{Lijun \surnamestart Zhang\surnameend} \&
  \bibinfo{author}{Jifeng \surnamestart He\surnameend} (\bibinfo{year}{2014}):
  \emph{\bibinfo{title}{Aalta: an {LTL} satisfiability checker over
  Infinite/Finite traces}}.
\newblock In \bibinfo{editor}{Shing{-}Chi \surnamestart Cheung\surnameend},
  \bibinfo{editor}{Alessandro \surnamestart Orso\surnameend} \&
  \bibinfo{editor}{Margaret{-}Anne~D. \surnamestart Storey\surnameend},
  editors: {\sl \bibinfo{booktitle}{Proceedings of the 22nd {ACM} {SIGSOFT}
  International Symposium on Foundations of Software Engineering, (FSE-22),
  Hong Kong, China, November 16 - 22, 2014}}, \bibinfo{publisher}{{ACM}}, pp.
  \bibinfo{pages}{731--734}, \doi{10.1145/2635868.2661669}.


\bibitemdeclare{article}{DBLP:journals/aicom/LudwigH10}
\bibitem{DBLP:journals/aicom/LudwigH10}
\bibinfo{author}{Michel \surnamestart Ludwig\surnameend} \&
  \bibinfo{author}{Ullrich \surnamestart Hustadt\surnameend}
  (\bibinfo{year}{2010}): \emph{\bibinfo{title}{Implementing a fair monodic
  temporal logic prover}}.
\newblock {\sl \bibinfo{journal}{AI Commun.}}
  \bibinfo{volume}{23}(\bibinfo{number}{2-3}), pp. \bibinfo{pages}{69--96}.
  \doi{10.3233/AIC-2010-0457}.

\bibitemdeclare{article}{DBLP:journals/corr/Reynolds16}
\bibitem{DBLP:journals/corr/Reynolds16}
\bibinfo{author}{Mark \surnamestart Reynolds\surnameend}
  (\bibinfo{year}{2016}): \emph{\bibinfo{title}{A traditional tree-style
  tableau for {LTL}}}.
\newblock {\sl \bibinfo{journal}{CoRR}} \bibinfo{volume}{abs/1604.03962}.
\newblock \urlprefix\url{https://arxiv.org/abs/1604.03962}.

\bibitemdeclare{inproceedings}{DBLP:conf/spin/RozierV07}
\bibitem{DBLP:conf/spin/RozierV07}
\bibinfo{author}{Kristin~Y. \surnamestart Rozier\surnameend} \&
  \bibinfo{author}{Moshe~Y. \surnamestart Vardi\surnameend}
  (\bibinfo{year}{2007}): \emph{\bibinfo{title}{LTL Satisfiability Checking}}.
\newblock In \bibinfo{editor}{Dragan \surnamestart Bosnacki\surnameend} \&
  \bibinfo{editor}{Stefan \surnamestart Edelkamp\surnameend}, editors: {\sl
  \bibinfo{booktitle}{SPIN}}, {\sl \bibinfo{series}{Lecture Notes in Computer
  Science}} \bibinfo{volume}{4595}, \bibinfo{publisher}{Springer}, pp.
  \bibinfo{pages}{149--167}.
  \doi{10.1007/978-3-540-73370-6_11}.

\bibitemdeclare{inproceedings}{RV11}
\bibitem{RV11}
\bibinfo{author}{Kristin~Y. \surnamestart Rozier\surnameend} \&
  \bibinfo{author}{Moshe~Y. \surnamestart Vardi\surnameend}
  (\bibinfo{year}{2011}): \emph{\bibinfo{title}{A Multi-encoding Approach for
  {LTL} Symbolic Satisfiability Checking}}.
\newblock In \bibinfo{editor}{Michael~J. \surnamestart Butler\surnameend} \&
  \bibinfo{editor}{Wolfram \surnamestart Schulte\surnameend}, editors: {\sl
  \bibinfo{booktitle}{{FM} 2011: Formal Methods - 17th International Symposium
  on Formal Methods, Limerick, Ireland, June 20-24, 2011. Proceedings}}, {\sl
  \bibinfo{series}{Lecture Notes in Computer Science}} \bibinfo{volume}{6664},
  \bibinfo{publisher}{Springer}, pp. \bibinfo{pages}{417--431},
  \doi{10.1007/978-3-642-21437-0_31}.


\bibitemdeclare{inproceedings}{VSchuppanLDarmawan-ATVA-2011}
\bibitem{VSchuppanLDarmawan-ATVA-2011}
\bibinfo{author}{Viktor \surnamestart Schuppan\surnameend} \&
  \bibinfo{author}{Luthfi \surnamestart Darmawan\surnameend}
  (\bibinfo{year}{2011}): \emph{\bibinfo{title}{Evaluating {LTL} Satisfiability
  Solvers}}.
\newblock In \bibinfo{editor}{Tevfik \surnamestart Bultan\surnameend} \&
  \bibinfo{editor}{Pao-Ann \surnamestart Hsiung\surnameend}, editors: {\sl
  \bibinfo{booktitle}{ATVA'11}}, {\sl \bibinfo{series}{Lecture Notes in
  Computer Science}} \bibinfo{volume}{6996}, \bibinfo{publisher}{Springer}, pp.
  \bibinfo{pages}{397--413}.
  \doi{10.1007/978-3-642-24372-1_28}.

\bibitemdeclare{phdthesis}{Sch98}
\bibitem{Sch98}
\bibinfo{author}{Stefan \surnamestart Schwendimann\surnameend}
  (\bibinfo{year}{1998}): \emph{\bibinfo{title}{Aspects of Computational
  Logic}}.
\newblock \bibinfo{type}{{PhD}}, \bibinfo{address}{Institut f{\"u}r Informatik
  und angewandte Mathematik}.
\newblock \urlprefix\url{http://www.iam.unibe.ch/ltgpub/1998/sch98b.ps}.

\bibitemdeclare{incollection}{Schwe98}
\bibitem{Schwe98}
\bibinfo{author}{Stefan \surnamestart Schwendimann\surnameend}
  (\bibinfo{year}{1998}): \emph{\bibinfo{title}{A New One-Pass Tableau Calculus
  for {PLTL}}}.
\newblock In \bibinfo{editor}{Harrie C.~M. \surnamestart de~Swart\surnameend},
  editor: {\sl \bibinfo{booktitle}{Automated Reasoning with Analytic Tableaux
  and Related Methods, International Conference, {TABLEAUX} '98, Oisterwijk,
  The Netherlands, May 5-8, 1998, Proceedings}}, {\sl \bibinfo{series}{Lecture
  Notes in Computer Science}} \bibinfo{volume}{1397},
  \bibinfo{publisher}{Springer}, pp. \bibinfo{pages}{277--292},
  \doi{10.1007/3-540-69778-0_28}.


\bibitemdeclare{book}{Smu68}
\bibitem{Smu68}
\bibinfo{author}{R.~\surnamestart Smullyan\surnameend} (\bibinfo{year}{1968}):
  \emph{\bibinfo{title}{First-order Logic}}.
\newblock \bibinfo{publisher}{Springer}.
  \doi{10.1007/978-3-642-86718-7}.

\bibitemdeclare{inproceedings}{SuW2012}
\bibitem{SuW2012}
\bibinfo{author}{Martin \surnamestart Suda\surnameend} \&
  \bibinfo{author}{Christoph \surnamestart Weidenbach\surnameend}
  (\bibinfo{year}{2012}): \emph{\bibinfo{title}{A PLTL-Prover Based on Labelled
  Superposition with Partial Model Guidance}}.
\newblock In \bibinfo{editor}{Bernhard \surnamestart Gramlich\surnameend},
  \bibinfo{editor}{Dale \surnamestart Miller\surnameend} \&
  \bibinfo{editor}{Uli \surnamestart Sattler\surnameend}, editors: {\sl
  \bibinfo{booktitle}{Automated Reasoning - 6th International Joint Conference,
  {IJCAR} 2012, Manchester, UK, June 26-29, 2012. Proceedings}}, {\sl
  \bibinfo{series}{Lecture Notes in Computer Science}} \bibinfo{volume}{7364},
  \bibinfo{publisher}{Springer}, pp. \bibinfo{pages}{537--543},
  \doi{10.1007/978-3-642-31365-3_42}.


\bibitemdeclare{article}{VaW94}
\bibitem{VaW94}
\bibinfo{author}{Moshe~Y. \surnamestart Vardi\surnameend} \&
  \bibinfo{author}{Pierre \surnamestart Wolper\surnameend}
  (\bibinfo{year}{1994}): \emph{\bibinfo{title}{Reasoning About Infinite
  Computations}}.
\newblock {\sl \bibinfo{journal}{Inf. Comput.}}
  \bibinfo{volume}{115}(\bibinfo{number}{1}), pp. \bibinfo{pages}{1--37},
  \doi{10.1006/inco.1994.1092}.


\bibitemdeclare{phdthesis}{Wid10}
\bibitem{Wid10}
\bibinfo{author}{Florian~Rainer \surnamestart Widmann\surnameend}
  (\bibinfo{year}{2010}): \emph{\bibinfo{title}{Tableaux-based decision
  procedures for fixed point logics}}.
\newblock \bibinfo{type}{Thesis}.
\newblock \urlprefix\url{http://users.cecs.anu.edu.au/~rpg/software.htm}.
\newblock \bibinfo{note}{Thesis (Ph.D.) -- ANU, 2010}.

\bibitemdeclare{article}{Wol85}
\bibitem{Wol85}
\bibinfo{author}{P.~\surnamestart Wolper\surnameend} (\bibinfo{year}{1985}):
  \emph{\bibinfo{title}{The tableau method for temporal logic: an overview}}.
\newblock {\sl \bibinfo{journal}{Logique et Analyse}} \bibinfo{volume}{28}, pp.
  \bibinfo{pages}{110--111}.

\end{thebibliography}

\end{document}